\documentclass[11pt]{article}       
\usepackage{geometry}               
\usepackage{graphicx}
\usepackage{caption}
\usepackage{subcaption}
\usepackage{comment}
\usepackage{amsmath} 
\usepackage{amssymb}  
\usepackage{amsthm}  
\usepackage{bm}
\usepackage{lipsum}
\usepackage{enumerate}
\usepackage{color}
\usepackage{cite}
\usepackage{mathrsfs}
\usepackage{hyperref}

\newtheorem{theorem}{Theorem}
\newtheorem{proposition}{Proposition}
\newtheorem{lemma}{Lemma}
\newtheorem{cor}{Corollary}

\newtheorem{assumption}{Assumption}

\numberwithin{equation}{section}

\DeclareMathOperator*{\argmax}{arg\,max}
\DeclareMathOperator*{\argmin}{arg\,min}

\title{Distributional Robustness in Minimax Linear Quadratic Control  
with Wasserstein Distance\thanks{
This work was supported in part by  the Creative-Pioneering Researchers Program through SNU,  the National Research Foundation of Korea funded by the MSIT(2020R1C1C1009766), and Samsung Electronics. } }

\author{Kihyun Kim \and
 Insoon Yang\thanks{Department of Electrical and Computer Engineering, Automation and Systems Research Institute,  Seoul National University, Seoul 08826, Korea {\tt\small \{hahakhkim, insoonyang\}@snu.ac.kr} }
}

\date{}

\begin{document}
\maketitle

\pagestyle{myheadings}
\thispagestyle{plain}

\begin{abstract}
To address the issue of inaccurate distributions in practical stochastic systems, a minimax linear-quadratic control method  is  proposed using the Wasserstein metric. 
Our method aims to construct a control policy that is robust against errors in an empirical distribution of underlying uncertainty, by adopting an adversary that selects the worst-case distribution. 
The opponent receives a \emph{Wasserstein penalty} proportional to the amount of deviation from the empirical distribution.
A closed-form expression of the finite-horizon optimal policy pair is derived using a Riccati equation. 
The result is then extended to the infinite-horizon average cost setting by identifying conditions under which the Riccati recursion converges to the unique positive semi-definite solution to an algebraic Riccati equation. 
Our method is shown to possess several salient features including closed-loop stability, and an out-of-sample performance guarantee. 
We also discuss how to optimize the penalty parameter for enhancing the distributional robustness of our control policy. 
Last but not least, a theoretical connection to the classical $H_\infty$-method is identified from the perspective of distributional robustness.  \end{abstract}

\section{Introduction}

Ambiguity, or uncertainty about uncertainty, in stochastic systems is one of the most fundamental challenges in the practical implementation of stochastic optimal controllers~\cite{Petersen2000, Tzortzis2015}. 
The true probability distribution of underlying uncertainty is unknown in ambiguous stochastic systems. 
In practice, we often only have access to samples generated according to the distribution. 
Estimating an accurate distribution from such observations is  challenging due to insufficient data and imperfect statistical models, among others. 
Using inaccurate distributions in the construction of an optimal policy may significantly decrease the control performance~\cite{Nilim2005, Samuelson2017} and can even cause unwanted system behaviors, such as unsafe operation~\cite{Yang2018aut}.
The focus of this work is to develop a discrete-time minimax control method using the Wasserstein metric and to analyze its robustness against uncertainties or errors in such distributional information.

Our work is closely related to the literature in \emph{distributionally robust control} (DRC).
DRC methods seek to design a control policy that minimizes an expected cost of interest under the worst-case distribution in a so-called \emph{ambiguity set}.
Several  types of ambiguity sets have been employed in DRC
 using moment constraints~\cite{Xu2012, VanParys2016}, confidence sets~\cite{Yang2017cdc}, relative entropy~\cite{Petersen2000, Ugrinovskii2002}, 
total variation distance~\cite{Tzortzis2015, Tzortzis2016}, and Wasserstein distance~\cite{Yang2017lcss, Yang2020}.\footnote{This paper focuses on  distributionally robust extensions of stochastic optimal control problems although distributionally robust techniques have also been studied in other control methods such as model predictive control~\cite{Coulson2019, Mark2020, Schuurmans2020, Ning2020}, and learning-based control~\cite{Schuurmans2019, Hakobyan2020}, among others.}
Such choices of ambiguity sets have largely been motivated by the literature in distributionally robust optimization (DRO)~\cite{Delage2010, BenTal2013, Wiesemann2014, Esfahani2015, Zhao2018, Gao2016}. 
In particular, DRO and DRC with the Wasserstein ambiguity set possess salient features such as a probabilistic out-of-sample performance guarantee and computational tractability~\cite{Esfahani2015, Zhao2018, Gao2016, Blanchet2018, Kuhn2019, Yang2020}.

In this paper, we propose a minimax linear-quadratic control method for ambiguous stochastic systems, inspired by Wasserstein DRC. 
To pursue distributional robustness, our method adopts a hypothetical opponent selecting the worst-case distribution to maximize a cost of interest, while the controller aims to minimize the same cost. 
To limit the conservativeness of the resulting control policy, 
our method penalizes the opponent by the amount (measured in the Wasserstein metric) of deviation from an empirical distribution.
 
The  minimax control problem is challenging to solve due to the infinite-dimensionality of the inner maximization problem in the Bellman equations. In the finite-horizon setting, we derive a Riccati equation and a closed-form expression of the unique optimal policy and the opponent's policy generating the worst-case distribution. 
In the infinite-horizon setting,  we identify a nontrivial stabilizability condition under which the solution to the Riccati equation converges to a symmetric positive semi-definite (PSD) solution to an algebraic Riccati equation (ARE). 
Taking a generalized eigenvalue approach,  our result is strengthened so that the converged solution corresponds to a unique symmetric PSD solution to the ARE under an additional observability condition. 
We also show that the resulting steady-state policy pair is an optimal solution to the infinite-horizon average cost minimax problem. 
The stability properties of the closed-loop system are further studied regarding the expected value of the system state.

We examine the distributional robustness of the resulting control policy, using Wasserstein ambiguity sets, motivated by the DRC formulation~\cite{Yang2020}.  
Specifically, we evaluate our policy under the worst-case distribution in the ambiguity set.
A simple upper-bound of this worst-case cost is derived using the optimal value function of our minimax problem. 
A penalty parameter minimizing the upper-bound can be computed by solving a convex optimization problem, which is obtained exploiting the structure property of the value function.
This study of our minimax method under a DRC lens yields another salient feature that our policy attains a performance guarantee evaluated under a new sample, independent of data used in the controller design. 
The probabilistic \emph{out-of-sample performance} guarantee is shown using the measure concentration inequality for the Wasserstein metric.

Another interesting observation is a theoretical connection between our minimax method and the $H_\infty$-method. 
Our method with Wasserstein distance can be understood as a distributional generalization of the $H_\infty$-method, thereby bridging the gap between stochastic and robust control. 
This connection yields the robust stability property of our minimax controller. 
Conversely, our stochastic interpretation of the $H_\infty$-method enables us to analyze the $H_\infty$-controller from the perspective of distributional robustness.

This paper is significantly expanded from its preliminary conference version~\cite{Kim2020}. The study of our minimax method using
the DRC formulation with a Wasserstein ambiguity set is newly presented along with the out-of-sample performance guarantee.  
Furthermore, the infinite-horizon total cost results in~\cite{Kim2020} are extended to the average cost setting, identifying optimality conditions and the guaranteed cost property.  
Last but not least, this paper contains the results regarding the bounded-input, bounded-output stability and the robust stability of the closed-loop system.

\section{Problem Formulations}\label{sec:setup}

Let $\mathbb{S}_+^n$ (resp. $\mathbb{S}_{++}^n$) denote the set of symmetric positive semi-definite (resp. positive definite) matrices in $\mathbb{R}^{n \times n}$.
Given a Borel set $\mathcal{W}$, let  $\mathcal{P} (\mathcal{W})$ denote the set of Borel probability measures on $\mathcal{W}$.
Moreover, $\| \cdot \|$ represents the standard Euclidean norm.

\subsection{Ambiguity in Stochastic Systems}

Consider a  discrete-time linear stochastic system of the form
\begin{equation} \label{sys}
    x_{t+1} = Ax_t + Bu_t + \Xi w_t, 
\end{equation}
where $x_t \in \mathbb{R}^n$ and $u_t \in \mathbb{R}^m$ represent the system state and input, respectively. 
Here, $w_t \in  \mathbb{R} ^k$ is a random disturbance vector with probability distribution $\mu_t \in \mathcal{P}(\mathbb{R}^k)$.
In addition,
$A \in\mathbb{R}^{n \times n}$, $B\in\mathbb{R}^{n \times m}$, and $\Xi\in\mathbb{R}^{n \times k}$ are time-invariant system matrices.

In practice, it is challenging to obtain the true probability distribution $\mu_t$ of $w_t$. 
One of the most straightforward ways to estimate the distribution is to construct the following empirical distribution from sample data $\{ \hat{w}^{(1)}_t, \ldots, \hat{w}^{(N)}_t\}$ of $w_t$:
\begin{equation}\label{emp}
\nu_t := \frac{1}{N} \sum_{i=1}^N \delta_{\hat{w}^{(i)}_t},
\end{equation}
where $\delta_{\hat{w}^{(i)}_t}$ denotes the Dirac measure concentrated at $\hat{w}^{(i)}_t$.
However, it is undesirable to use this empirical distribution in controller design because the control performance would deteriorate as the true distribution deviates from $\nu_t$. 

\subsection{Minimax Control with Wasserstein Penalty}
Let $\pi := (\pi_0, \pi_1, \ldots)$ denote a deterministic Markov control policy, where $\pi_t$ maps the current state $x_t$ to 
an input $u_t$.\footnote{For ease of exposition, we focus on deterministic Markov policies. However, all the results in this paper are valid even when considering randomized history-dependent policies for both players by the optimality result in~\cite{Yang2020}.}
More precisely, the set of admissible control policies is given by
$\Pi:= \{ \pi \mid \pi_t (x_t) = u_t \in \mathbb{R}^m, \: \pi_t \mbox{ is measurable} \; \forall t\}$.
To design a controller that is robust against errors in the empirical distributions, 
we employ an (hypothetical) opponent that selects the probability distribution $\mu_t$ in an adversarial way.
The opponent policy $\gamma := (\gamma_0, \gamma_1, \ldots)$ is also assumed to be deterministic and Markov, where $\gamma_t$ maps the current state $x_t$ to a probability distribution $\mu_t$.
Specifically, the set of admissible opponent's policies is defined by
$\Gamma := \{\gamma \mid \gamma_t (x_t) = \mu_t \in \mathcal{P}(\mathbb{R}^k) \; \forall t\}$.
We first consider the finite-horizon case and later extend our results to the infinite-horizon case.

Suppose for a moment that the controller aims to minimize the standard quadratic cost function 
\begin{equation} \label{cost0}
\begin{split}
&J_{\bm{x}}(\pi, \gamma) = J_{\bm{x}, T}(\pi, \gamma) := \frac{1}{T}\mathbb{E}^{\pi, \gamma} \bigg [x_T^\top Q_f x_T 
+ \sum_{t=0}^{T-1} \big ( x_t^\top Q x_t + u_t^\top R u_t  \big )~\bigg\vert~x_0 = \bm{x}
\bigg ],
\end{split}
\end{equation}
 with $Q, Q_f \in \mathbb{S}_+^n$ and $R \in \mathbb{S}_{++}^m$, while the opponent determines $\gamma$ to maximize the same cost. 
If this were the case,
however, that would give too much freedom to the opponent, thereby causing the optimal controller to be overly conservative.  
To systematically adjust conservativeness, we penalize the opponent according to the degree of deviation from the empirical distributions $\nu_t$'s. 
By doing so, we can also incorporate the prior information provided by the sample data directly into the controller design. 
Specifically, the penalty is measured by the Wasserstein distance $W_2(\mu_t, \nu_t)$ between $\mu_t$ and $\nu_t$.
The Wasserstein metric of order $2$ between two distributions $\mu$ and $\nu$ is defined as
\begin{equation*}
\begin{split}
     W_2(\mu,\nu):= \inf_{\eta\in \mathcal{P}(\mathcal{W}^2)}
     \bigg\{ \left ( \int_{\mathcal{W}^2} \|x - y\|^2 \mathrm{d}\eta(x,y) \right ) ^{\frac{1}{2}}~\bigg \vert~\Pi^1\eta=\mu, \Pi^2 \eta=\nu \bigg\},  
\end{split}
\end{equation*}
where $\Pi^i \eta$ is the $i$th marginal distribution of $\eta$.
The cost function is then modified by adding a Wasserstein penalty term as follows:
\begin{equation} \label{cost1}
\begin{split}
&J^\lambda_{\bm{x}}(\pi, \gamma) = J_{\bm{x}, T}^\lambda (\pi, \gamma)  := \frac{1}{T}\mathbb{E}^{\pi, \gamma} \bigg [x_T^\top Q_f x_T 
+ \sum_{t=0}^{T-1} \big ( x_t^\top Q x_t + u_t^\top R u_t - \lambda W_2 (\mu_t, \nu_t)^2 \big )~\bigg\vert~x_0 = \bm{x}
\bigg ],
\end{split}
\end{equation}
where $\lambda >0$ is the penalty parameter. By definition, $J_{\bm{x}} (\pi, \gamma) = J_{\bm{x}}^0 (\pi, \gamma)$.
Tuning the parameter $\lambda$, we can adjust the conservativeness of our control policy that is obtained by solving the following minimax stochastic control problem:
\begin{equation} \label{opt}
\min_{\pi \in \Pi} \max_{\gamma \in \Gamma} J^\lambda_{\bm{x}} (\pi, \gamma).
\end{equation}
The inner maximization problem yields a worst-case distribution policy given $\pi$. 
Thus, an optimal solution $\pi^\star$ to the outer problem minimizes the worst-case cost.
Our first goal is to develop a Riccati equation-based solution to \eqref{opt} and analyze the properties of $\pi^\star$ such as closed-loop stability.

\subsection{Distributional Robustness}\label{sec:dual}

  A closely related minimax stochastic control formulation is the distributionally robust control problem~\cite{Yang2020}.
 This formulation uses Wasserstein ambiguity sets instead of the Wasserstein penalty term. 
Specifically, the Wasserstein ambiguity set  is defined as
\begin{equation}\label{eq:const}
\begin{split}
    \mathcal{D}_t &:= \{ \mu_t \in \mathcal{P}(\mathbb{R}^{k}) ~\vert~W_2(\mu_t, \nu_t) \leq \theta \}.
\end{split}
\end{equation}
The set $\mathcal{D}_t$ is a  statistical ball centered at the empirical distribution $\nu_t$, where the distance between any two elements  is measured by the Wasserstein metric. 
The opponent's policy $\gamma$ is then be restricted in the following space:
\begin{equation*}
\begin{split}
    \Gamma_\mathcal{D} &:= \{ \gamma \in \Gamma~\vert~  \gamma_t(x_t) \in \mathcal{D}_t \; \forall t    \}.
\end{split}
\end{equation*}
In words, the probability distribution produced by the opponent's policy must be contained in the Wasserstein ambiguity set. 
To achieve distributional robustness,
 it is desirable to design a controller that minimizes the expected cost under the worst-case distribution policy in $\Gamma_\mathcal{D}$.
Such a control policy can be obtained by solving the following Wasserstein distributionally robust control problem: 
\begin{equation}\label{opt2}
    \min_{\pi \in \Pi} \max_{\gamma \in \Gamma_\mathcal{D}} J_{\bm{x}} (\pi, \gamma),
\end{equation}
which can be solved by dynamic programming (DP). 
Unfortunately, the DP solution is not scalable due to the curse of dimensionality unlike our Riccati equation-based method. 
We claim that the optimal policy $\pi^\star$ of \eqref{opt} is a reasonable suboptimal solution to the DR control problem since
it has the following guaranteed-cost property:
\begin{equation}\label{eq:reform}
\sup_{\gamma \in \Gamma_\mathcal{D}} J_{\bm{x}} (\pi^\star(\lambda^\star), \gamma) 
    \leq \lambda^\star \theta^2 +    V (\bm{x}; \lambda^\star),
\end{equation}
where 
$\pi^\star (\lambda)$ denotes the optimal policy of \eqref{opt} with $\lambda$,
\[
V(\bm{x}; \lambda) := \inf_{\pi \in \Pi} \sup_{\gamma \in \Gamma} J^\lambda_{\bm{x}} (\pi, \gamma)
\]
denotes the optimal value function of \eqref{opt}, and
$\lambda^\star \in \argmin_{\lambda \geq 0} [ \lambda \theta^2 + V(\bm{x}; \lambda)]$
Note that the objective function of the minimization problem on the right-hand side can be evaluated by solving~\eqref{opt}.
Thus, the right-hand side provides a provable upper-bound on the worst-case cost of employing $\pi^\star(\lambda^\star)$. 
This upper-bound can be used to speculate the distributional robustness of $\pi^\star(\lambda^\star)$ and to quantify a probabilistic out-of-sample performance guarantee of $\pi^\star(\lambda^\star)$ as discussed in Section~\ref{sec:dr}.
In the following section, we first study the problem~\eqref{opt} to obtain an explicit solution in both finite-horizon and infinite-horizon cases and identify useful properties.
These results will then be used to analyze the distributional robustness of $\pi^\star(\lambda^\star)$ in Section~\ref{sec:dr}.

\section{Minimax Linear Quadratic Control with Wasserstein Penalty}\label{sec:penalty}
\subsection{Finite-Horizon Case}\label{sec:finite}

To begin with, we consider the regularized problem~\eqref{opt} in the finite-horizon setting with cost function $J^\lambda_{\bm{x}} (\pi, \gamma)$, defined in \eqref{cost1}. Later, we establish the connection between the finite-horizon and infinite-horizon cases by letting $T \to \infty$.

We use dynamic programming to solve the finite-horizon problem: 
let the optimal value function be defined by
$V_t(\bm{x}) = V_t(\bm{x}; \lambda) :=\inf_{\pi \in \Pi} \sup_{\gamma \in \Gamma}   \mathbb{E}^{\pi, \gamma} [ \sum_{s=t}^{T-1} (x_s^\top Q x_s + u_s^\top R u_s    -\lambda W_2(\mu_s, \nu_s)^2) 
+ x_T^\top Q_f x_T \mid x_t = \bm{x} ]$,
which represents the optimal worst-case expected cost-to-go from stage $t$ given $x_t = \bm{x}$.
By definition, $V (\bm{x}; \lambda) = V_0 (\bm{x}; \lambda) / T$.
The dynamic programming principle yields
\begin{equation}\nonumber
\begin{split}
    V_t (\bm{x}) = \bm{x}^\top Q \bm{x} &+ \inf_{\bm{u} \in \mathbb{R}^m} \sup_{\bm{\mu} \in \mathcal{P}(\mathbb{R}^k)} \bigg [ \bm{u}^\top R \bm{u} - \lambda W_2(\bm{\mu}, \nu_t)^2 +\int_{\mathbb{R}^k} V_{t+1} (A \bm{x} + B \bm{u} + \Xi w) \mathrm{d} \bm{\mu} (w) \bigg ]
\end{split}
\end{equation}
for $t = 0, \ldots, T-1$, and $V_T(\bm{x}) := \bm{x}^\top Q_f \bm{x}$.
Note that the inner maximization problem is an infinite-dimensional optimization problem over $\mathcal{P}(\mathbb{R}^k)$.
For a tractable reformulation, we use a modern DRO technique based on Kantorovich duality~\cite{Gao2016}, which yields
\begin{equation}\label{eq:backward}
\begin{split}
    &V_t (\bm{x}) =  \bm{x}^\top  Q \bm{x} + \inf_{\bm{u} \in \mathbb{R}^m} \bigg[ \bm{u}^\top R \bm{u}  +   \frac{1}{N} \sum_{i=1}^N \sup_{ w\in \mathbb{R}^k} \big  \{ V_{t+1}(A \bm{x} 
    + B \bm{u} + \Xi w)-\lambda \lVert \hat{w}^{(i)}_t - w \rVert^2 \big  \} \bigg].
\end{split}
\end{equation}
Let the mean and the covariance matrix of the empirical distribution be denoted by
\[
\bar{w}_t:= \mathbb{E}_{\nu_t} [w_t], \quad \Sigma_t:=\mathbb{E}_{\nu_t} [w_t w_t^\top].
\]
We also let
\begin{equation}\label{phi}
\Phi := BR^{-1} B^\top - \frac{1}{\lambda} \Xi \Xi^\top.
\end{equation}

We now consider the following ansatz of the value function: 
 $V_t (\bm{x}) = \bm{x}^\top P_t \bm{x} + 2 r_t^\top \bm{x} + z_t$, where $P_t  \in \mathbb{S}_+^{n}$, $r_t \in \mathbb{R}^n$ and $z_{t} \in \mathbb{R}$.
 Our goal is to identify an explicit solution to the minimax optimization problem in \eqref{eq:backward}. 
In what follows, we show that the quadratic structure of the value function is preserved through the Bellman recursion, and  the proposed parameterization would thus be exact if  matrices $P_t$ satisfy a Riccati equation.  

\begin{lemma}\label{lem:sol}
Suppose that 
\[
V_{t+1} (\bm{x}) = \bm{x}^\top P_{t+1} \bm{x} + 2r_{t+1}^\top \bm{x} + z_{t+1}
\]
 for some $P_{t+1}   \in \mathbb{S}_+^{n}$, $r_{t+1} \in \mathbb{R}^n$ and $z_{t+1} \in \mathbb{R}$. 
We further assume that the penalty parameter satisfies $\lambda > \bar{\lambda}_{t+1}$, where $\bar{\lambda}_{t+1}$ is the maximum eigenvalue of $\Xi^\top P_{t+1} \Xi$. 
Then, the inner maximization problem 
$\sup_{ w\in \mathbb{R}^k} \{ V_{t+1}(A \bm{x} 
    + B \bm{u} + \Xi w)-\lambda \lVert \hat{w}^{(i)}_t - w \rVert^2 \}$
in \eqref{eq:backward} has a unique maximizer $w_t^\star := (w^{\star,(1)}_{t}, \ldots, w^{\star,(N)}_{t})$, defined as
\begin{equation}\label{w_opt}
\begin{split}
w^{\star,(i)}_{t} &:= (\lambda I - \Xi^\top P_{t+1} \Xi)^{-1} (\Xi^\top P_{t+1} (A\bm{x} + B \bm{u}) + \Xi^\top r_{t+1} + \lambda \hat{w}_t^{(i)}).
\end{split}
\end{equation}
Furthermore, the outer minimization problem in \eqref{eq:backward} has a unique minimizer 
\begin{equation}\label{u_opt}
\begin{split}
\bm{u}^\star := K_t \bm{x} + L_t, 
\end{split}
\end{equation}
where
\begin{equation*}
\begin{split}
&K_t :=  -R^{-1} B^\top  (I + P_{t+1} \Phi)^{-1} P_{t+1} A,\\
&L_t :=  -R^{-1} B^\top  (I + P_{t+1} \Phi)^{-1} ( P_{t+1}\Xi \bar{w}_t + r_{t+1}).
\end{split}
\end{equation*}
\end{lemma}
 
Note that $w_t^\star$ is linear in $(\bm{x}, \bm{u})$ and $\bm{u}^\star$ is linear in $\bm{x}$. 
The explicit derivation with this linear structure yields the following Riccati equation:
\begin{equation}\label{ric}
\begin{split}
P_t &= Q + A^\top (I + P_{t+1}\Phi )^{-1} P_{t+1} A\\
r_t &=  A^\top (I + P_{t+1} \Phi)^{-1} ( P_{t+1}\Xi \bar{w}_t + r_{t+1} )\\
z_t &= z_{t+1}  + \mathrm{tr}   [ (I -  \Xi^\top P_{t+1} \Xi/\lambda )^{-1} \Xi^\top P_{t+1} \Xi \Sigma_t   ]\\
&+\bar{w}_t^\top \Xi^\top  [(I+P_{t+1}\Phi)^{-1}-(I- P_{t+1} \Xi \Xi^\top/\lambda)^{-1}] P_{t+1} \Xi \bar{w}_t\\
&  + (2  \bar{w}_t^\top \Xi^\top - r_{t+1}^\top \Phi)(I+P_{t+1} \Phi)^{-1} r_{t+1}
\end{split}
\end{equation}
with the terminal conditions $P_T  = Q_f$, $r_T =0$, and $z_T = 0$.
Note that $P_t$, $t = 0, \ldots, T-1$, are symmetric since $P_T$ is symmetric.
For the well-definedness of the recursion, we make the following assumption:
\begin{assumption}\label{ass:pen}
The penalty parameter satisfies $\lambda > \bar{\lambda}_{t}$ for all~$t\geq1$, where $\bar{\lambda}_{t}$ is the maximum eigenvalue of $\Xi^\top P_{t} \Xi$.
\end{assumption}

\begin{theorem}[Optimal policy]\label{thm:fin} 
Suppose that Assumption~\ref{ass:pen} holds.
Then, the matrices $P_t$ are well-defined and 
the value function can be expressed as
\[
V_t(\bm{x}) = \bm{x}^\top P_t \bm{x} + 2 r_t^\top \bm{x} +  z_t, \quad t = 0, \ldots, T. 
\]
Furthermore, the regularized problem~\eqref{opt} in the finite-horizon case has a unique optimal policy, defined as
\begin{equation} \label{opt_policy}
\pi^\star_t (\bm{x}) := K_t \bm{x} + L_t, \quad t = 0, \ldots, T-1.
\end{equation}
\end{theorem}
 
As in the standard LQG, the optimal policy is linear in system state and gain matrix $K_t$ can be obtained by solving a Riccati equation. 
Note that the Riccati equation in the standard LQG is given by (e.g.,~\cite{Astrom2012})
\begin{equation}\label{ric_lqg}
\begin{split}
P_t &= Q + A^\top   (I + P_{t+1} B R^{-1}B^\top )^{-1} P_{t+1} A\\
 r_t &=  A^\top   (I + P_{t+1} B R^{-1}B^\top )^{-1} ( P_{t+1}\Xi \bar{w}_t + r_{t+1})\\
z_t &= z_{t+1} + \mathrm{tr}[ \Xi^\top P_{t+1} \Xi \Sigma_t ]\\
&- \bar{w}_t^\top \Xi^\top  P_{t+1} B R^{-1}B^\top (I + P_{t+1} B R^{-1}B^\top )^{-1} P_{t+1} \Xi \bar{w}_t \\
& + (2  \bar{w}_t^\top \Xi^\top - r_{t+1}^\top  B R^{-1}B^\top)(I+P_{t+1} B R^{-1}B^\top)^{-1} r_{t+1},
\end{split}
\end{equation}
and it can be obtained by letting $\lambda \to \infty$ in \eqref{ric}.
Increasing $\lambda$ encourages the opponent not to deviate much from the empirical distribution $\nu_t$. Thus, in the limit,  our minimax method is equivalent to the standard LQG. 
This shows that our proposed framework is a generalization of LQG.

Another immediate consequence of Lemma~\ref{lem:sol} and Theorem~\ref{thm:fin} is that one of the worst-case distributions can be explicitly obtained with a finite support, as follows:

\begin{cor}[Worst-case distribution]\label{cor:dist}
Suppose that Assumption~\ref{ass:pen} holds. 
Let 
\begin{equation} \nonumber
\begin{split}
&w_t^{\star, (i)} (\bm{x}) := (\lambda I - \Xi^\top P_{t+1} \Xi)^{-1}  (\Xi^\top P_{t+1} (A \bm{x} + BK_t \bm{x} + BL_t) + \Xi^\top r_{t+1} + \lambda \hat{w}_t^{(i)}).
\end{split}
\end{equation}
Then, the policy $\gamma^\star$ defined by
\[
\gamma^\star_t (\bm{x}) := \frac{1}{N} \sum_{i=1}^N \delta_{w_t^{\star, (i)} (\bm{x})} 
\]
generates the worst-case distribution, i.e., $(\pi^\star, \gamma^\star)$ is an optimal minimax solution to \eqref{opt} in the finite-horizon case.
\end{cor}

\subsection{Infinite-Horizon Case}\label{sec:infinite}

In this subsection, we investigate an optimal controller for the infinite-horizon case when the number of stage $T$ increases to $\infty$.
We consider the following infinite-horizon average cost criterion:
\begin{equation}\label{ac_cost}
\begin{split}
J^{\lambda}_{\bm{x},\infty}(\pi, \gamma) := \limsup_{T \rightarrow \infty} &\frac{1}{T} \mathbb{E}^{\pi, \gamma} \bigg[ \sum_{t=0}^{T-1} ( x_t^\top Q x_t 
+ u_t^\top R u_t - \lambda W_2(\mu_t, \nu)^2)~\bigg\vert~x_0 = \bm{x} \bigg].
\end{split}
\end{equation}
Based on the results in the finite-horizon case, we begin by identifying the steady-state policy that our optimal policy converges to.
Our specific goal is to derive an algebraic Riccati equation (ARE) and characterize the condition under which the recursion~\eqref{ric} converges to 
 a unique symmetric PSD solution of the ARE. 

Throughout this subsection, we assume the following for the stationarity of the problem.
\begin{assumption}\label{ass:st}
The random disturbance process $\{w_t\}_{t=0}^\infty$ is i.i.d.,
and its empirical distribution is constructed as
$\nu \equiv \nu_t := \frac{1}{N} \sum_{i=1}^N \delta_{\hat{w}^{(i)}}$ from the dataset $\{ \hat{w}^{(1)}, \ldots, \hat{w}^{(N)} \}$.
\end{assumption}
Under Assumption~\ref{ass:st}, we denote the mean value and the covariance matrix of $\nu$ by $\bar{w}$ and $\Sigma$.
Based on the iteration~\eqref{ric}, our focus is on finding a solution to the following algebraic Riccati equation (ARE):
\begin{equation}\label{are}
P = Q + A^\top  \bigg [ I + PBR^{-1} B^\top - \frac{1}{\lambda} P \Xi \Xi^\top \bigg ]^{-1} PA.
\end{equation}
Note that the ARE \eqref{are} has an equivalent form to the ARE in the classical $H_\infty$-optimal control (see \cite{Basar2008}[Section 3.2]). The specific relationship between our minimax method and the $H_\infty$-method will be discussed in the Section~\ref{sec:Hinfinity}.

\subsubsection{Algebraic Riccati Equation}

We first show that 
$P_t$ updated by~\eqref{ric} converges to a unique PSD solution of the ARE~\eqref{are} under suitable nontrivial stabilizability and observability conditions.
Recall that the symmetric matrix $\Phi$ is defined as \eqref{phi}.
We make the following assumption on $\Phi$:

\begin{assumption} \label{ass:W}
$\Phi \succeq 0$, and $(A, {\Phi}^{1/2})$ is stabilizable.
\end{assumption}

\begin{proposition}\label{prop:are}
Suppose that Assumptions~\ref{ass:pen}--\ref{ass:W} hold. 
Then, a bounded limiting solution $P_{ss}:= \lim_{T \to \infty} P_t$ to the Riccati equation~\eqref{ric} exists for any  $P_T \in \mathbb{S}_+^n$.
Furthermore, $P_{ss}$ is a symmetric PSD solution to the ARE \eqref{are}. 
\end{proposition}

To solve the ARE~\eqref{are}, we use the method proposed in~\cite{Pappas1980}, considering the generalized eigenvalue problem of $F$ and $G$
\begin{equation}\label{gen}
F v = \gamma G v,
\end{equation}
where $F := \begin{bmatrix} A & 0 \\ -Q & I\end{bmatrix}$ and $G :=\begin{bmatrix} I & \Phi \\ 0 & A^\top\end{bmatrix}$.

\begin{lemma}\label{lem:sol2}
Any solution of the ARE~\eqref{are} can be expressed as
\[
P = U_2 U_1^{-1},
\]
where each column of $\begin{bmatrix} U_1 \\ U_2 \end{bmatrix} \in \mathbb{R}^{2n \times n}$ solves the generalized eigenvalue problem~\eqref{gen} of $F$ and $G$.
\end{lemma}
 
Lemma~\ref{lem:sol2} shows that all solutions of the ARE~\eqref{are} can be obtained from the generalized eigenvalue problem of $F$ and $G$. 
Unfortunately, most of them are unstabilizing solutions.
However, we are only interested in the symmetric PSD solution $P_{ss}$ to which the Riccati recursion~\eqref{ric} converges. 
To identify the steady-state solution, we need the following assumption and lemma: 

\begin{assumption}\label{ass:ob}
$(A,  {Q}^{1/2})$ is observable. 
\end{assumption}

\begin{lemma}\label{lem:stable}
Suppose that Assumptions~\ref{ass:W} and \ref{ass:ob} hold.
Then, $P = U_2 U_1^{-1}$ is a symmetric PSD solution to the ARE~\eqref{are} if and only if each column of  $\begin{bmatrix} U_1 \\ U_2 \end{bmatrix} \in \mathbb{R}^{2n \times n}$ solves the generalized eigenvalue problem~\eqref{gen} of $F$ and $G$ with a stable generalized eigenvalue.\footnote{A generalized eigenvalue is stable if its absolute value is less than $1$.} 
\end{lemma}

Lemma~\ref{lem:stable} motivates us to investigate the condition on $F$ and $G$ under which \eqref{gen} has $n$ stable generalized eigenvalues. 
Note that  the following symplectic property holds
\[
F \Omega F^\top = G \Omega G^\top  = \begin{bmatrix} 0 & A \\ -A^\top & 0\end{bmatrix},
\]
where $\Omega = \begin{bmatrix} 0 & I_n \\ -I_n & 0\end{bmatrix}$.
Thus, if $\gamma$ is a generalized eigenvalue, so is ${1}/{\gamma}$ with the same multiplicity.
This implies that if no generalized eigenvalue lies on the unit circle, then exactly $n$ generalized eigenvalues are stable, and there exists a unique symmetric PSD solution to the ARE by Lemma~\ref{lem:stable}.

\begin{lemma}\label{lem:circ}
Under Assumptions~\ref{ass:W} and \ref{ass:ob}, $(F, G)$ does not have any generalized eigenvalue on the unit circle. 
\end{lemma}
\begin{proof}
The existence of generalized eigenvalues on the unit circle contradicts Assumptions \ref{ass:W} and \ref{ass:ob}. See \cite[Theorem 3]{Pappas1980} for details.
\end{proof}

By Lemma~\ref{lem:circ}, there exist $U_1, U_2 \in \mathbb{R}^{n\times n}$ and $\Lambda \in \mathbb{R}^{n\times n}$ such that 
\begin{equation}\label{gen2}
FU=GU \Lambda
\end{equation}
 with $U=\begin{bmatrix} U_{1}\\ U_{2}\end{bmatrix}$, where 
the  columns of $V$ solve \eqref{gen} with $n$ stable generalized eigenvalues, and 
$\Lambda$ is the corresponding Jordan normal form. 
We obtain the following lemma that yields to construct a solution of the ARE~\eqref{are} from $U_1$ and $U_2$. 

\begin{lemma}
Under Assumptions~\ref{ass:W} and \ref{ass:ob},  $U_1$ is nonsingular. 
\end{lemma}
\begin{proof}
This can be shown directly using the proof of \cite[Theorem 6]{Pappas1980}.
\end{proof}

Using the previous lemmas, we finally obtain the following conclusion that connects the Riccati equation~\eqref{ric} in the finite-horizon case and the ARE~\eqref{are} in the infinite-horizon. 

\begin{theorem}
Suppose that Assumptions~\ref{ass:pen}--\ref{ass:ob} hold. 
Then, the recursion~\eqref{ric} converges to the unique symmetric PSD solution $P_{ss} := U_2 U_1^{-1}$ of the ARE~\eqref{are} as $T \to \infty$.
\end{theorem}

This result can further be simplified 
when the system matrix $A$ is nonsingular.
In this particular case,  we let
\begin{equation*}
    H := G^{-1}F =\begin{bmatrix} A + \Phi A^{-\top} Q & - \Phi A^{-\top} \\ - A^{-\top} Q & A^{-\top} \end{bmatrix}
\end{equation*}
and construct $\hat{U}_1, \hat{U}_2 \in \mathbb{R}^{n \times n}$ so that each column of $\begin{bmatrix} \hat{U}_{1} \\ \hat{U}_{2} \end{bmatrix} \in \mathbb{R}^{2n \times n}$ is an eigenvector of $H$ associated with a stable eigenvalue. 
We then obtain the following result:

\begin{cor}
Suppose that Assumptions~\ref{ass:pen}--\ref{ass:ob} hold and that $A$ is nonsingular. 
Then, the recursion~\eqref{ric} converges to the unique symmetric PSD solution $P_{ss} := \hat{U}_2 \hat{U}_1^{-1}$ of the ARE~\eqref{are}.
\end{cor}

The convergence of $r_t$ in the recursion \eqref{ric} directly follows from the convergence of $P_t$.
\begin{proposition} \label{prop:r}
Suppose that Assumptions~\ref{ass:pen}--\ref{ass:ob} hold. 
Then, $r_t$ in the recursion~\eqref{ric} converges to
\begin{equation}\label{eq:oss}
\begin{split}
r_{ss} := &[I - A^\top (I+P_{ss} \Phi)^{-1}]^{-1} A^\top (I+P_{ss} \Phi)^{-1} P_{ss} \Xi \bar{w}
\end{split}
\end{equation}
as $T \to \infty$.
\end{proposition}

The steady-state control policy in the infinite-horizon case
can be obtained using the symmetric PSD solution to the ARE~\eqref{are} as in the finite-horizon case.

\begin{cor}\label{cor:inf}
Suppose that Assumptions~\ref{ass:pen}--\ref{ass:ob} hold.
Then, the optimal policy $\pi^\star_t(\bm{x})$ converges to the steady-state policy
\[
\pi^\star_{ss} (\bm{x}) := K_{ss} \bm{x} + L_{ss}
\]
as $T \to \infty$, where
\begin{equation*}
\begin{split}
K_{ss} := & -R^{-1} B^\top  (I + P_{ss} \Phi)^{-1} P_{ss} A,\\
L_{ss} := & -R^{-1} B^\top (I + P_{ss} \Phi)^{-1}( P_{ss}\Xi \bar{w} + r_{ss}).
\end{split}
\end{equation*}
Furthermore, the policy generating the worst-case distribution converges to the steady-state policy
\begin{equation}\label{eq:steady}
\gamma^\star_{ss} (\bm{x}) := \frac{1}{N} \sum_{i=1}^N \delta_{w^{\star, (i)} (\bm{x})} 
\end{equation}
with
$w^{\star, (i)} (\bm{x}) 
:= (\lambda I - \Xi^\top P_{ss} \Xi)^{-1} (\Xi^\top P_{ss} (A \bm{x} + BK_{ss} \bm{x} + BL_{ss}) + \Xi^\top r_{ss} + \lambda \hat{w}^{(i)})$.
\end{cor}

\begin{proposition}\label{prop:avgcost}
Suppose that Assumptions~\ref{ass:pen}--\ref{ass:ob} hold.
Then, the steady-state average cost 
\begin{equation*}
\rho := \limsup_{T \rightarrow \infty}  \min_{\pi \in \Pi} \max_{\gamma \in \Gamma} J^\lambda_{\bm{x}, T} (\pi, \gamma)
\end{equation*}
is given by
\begin{equation*}
\begin{split}
	\rho &= \mathrm{tr} [  (I -\Xi^\top P_{ss} \Xi /\lambda )^{-1} \Xi^\top P_{ss} \Xi \Sigma ]+ \bar{w}^\top \Xi^\top  [(I+P_{ss} \Phi )^{-1}-(I- P_{ss} \Xi \Xi^\top/\lambda)^{-1}] P_{ss} \Xi \bar{w} \\
 &+ (2  \bar{w}^\top \Xi^\top - r_{ss}^\top \Phi)(I+P_{ss} \Phi )^{-1} r_{ss},
\end{split}
\end{equation*}
independent of the initial state $\bm{x}$.
\end{proposition}

\subsubsection{Average Cost Optimality}

We now examine  the optimality of the stationary policy pair $(\pi^\star_{ss}, \gamma^\star_{ss})$ using the average cost criterion~\eqref{ac_cost}.
Consider the following average cost problem with a Wasserstein penalty:\footnote{It follows from the definition of $\Pi$ that $(\pi_{ss}^\star, \pi_{ss}^\star, \ldots) \in \Pi$.
However, with a slight abuse of notation, we simply denote it as $\pi_{ss}^\star$ and 
regard stationary policy $\pi_{ss}^\star$ (resp. $\gamma_{ss}^\star$)  as an element of  $\Pi$ (resp. $\Gamma$). }
\begin{equation}
	\min_{\pi \in \Pi} \max_{\gamma \in \Gamma} J^{\lambda}_{\bm{x}, \infty} (\pi, \gamma).
\end{equation}
The optimality equation for this problem can be obtained as follows:

\begin{proposition}\label{prop:acoe}
Suppose that Assumptions~\ref{ass:pen}--\ref{ass:ob} hold.
Then, the following Bellman equation holds:
\begin{equation} \label{eq:acoe}
\begin{split}
   & \rho + h (\bm{x}) = \bm{x}^\top Q \bm{x} + \inf_{\bm{u} \in \mathbb{R}^m} \sup_{\bm{\mu} \in \mathcal{P}(\mathbb{R}^k)} \bigg [ \bm{u}^\top R \bm{u}   - \lambda W_2(\bm{\mu}, \nu)^2 +\int_{\mathbb{R}^k} h (A \bm{x} + B \bm{u} + \Xi w) \mathrm{d} \bm{\mu} (w) \bigg ],
\end{split}
\end{equation}
where $h(\bm{x}) := \bm{x}^\top P_{ss} \bm{x} + 2r_{ss}^\top \bm{x}$ and $\rho$ is the steady-state average cost defined in Proposition~\ref{prop:avgcost}.
Moreover, $(\pi_{ss}^\star (\bm{x}), \gamma_{ss}^\star (\bm{x}))$ is an optimal minimax solution of the problem on the right-hand side of~\eqref{eq:acoe}. 
\end{proposition}

In the Bellman equation (or the average cost optimality equation), $h$, called the \emph{bias},  represents the transient cost, whereas $\rho$, called the \emph{gain}, represents the stationary cost.

We now introduce an extended average cost function including the bias $h$ as
\begin{equation}\nonumber
\begin{split}
\tilde{J}^{\lambda}_{\bm{x}, \infty}(\pi, \gamma;h) := \limsup_{T \rightarrow \infty} \frac{1}{T} \mathbb{E}^{\pi, \gamma} \bigg[ \sum_{t=0}^{T-1} ( x_t^\top Q x_t 
+ u_t^\top R u_t - \lambda W_2(\mu_t, \nu)^2 )+ h(x_T)~\bigg\vert~x_0 = \bm{x} \bigg].
\end{split}
\end{equation}
Using the extended cost, we can show  the average cost optimality of the stationary policy pair $(\pi_{ss}^\star, \gamma_{ss}^\star)$ in a way similar to the average cost LQG (e.g., \cite{Bertsekas2012}[Section 5.6.5]).

\begin{theorem}\label{thm:inf}
Suppose that Assumptions~\ref{ass:pen}--\ref{ass:ob} hold.
Consider the steady-state policy pair $(\pi^\star_{ss}, \gamma^\star_{ss})$ defined in Corollary~\ref{cor:inf}.
Then, 
the following properties hold:
\begin{enumerate}[(a)]

\item 
For any $(\pi, \gamma) \in \Pi \times \Gamma$,
 \begin{equation}\label{eq:saddle}
\tilde{J}^{\lambda}_{\bm{x}, \infty}(\pi_{ss}^\star, \gamma;h)  \leq \rho
 \leq \tilde{J}^{\lambda}_{\bm{x}, \infty}(\pi, \gamma_{ss}^\star;h),
\end{equation}
where $\rho$ is the stationary cost defined in Proposition~\ref{prop:avgcost}.

\item  The stationary policy pair $(\pi_{ss}^\star, \gamma_{ss}^\star)$ is optimal to 
\begin{equation}\label{eq:avgopt}
	\min_{\pi \in \Pi}\max_{\gamma \in \Gamma} \;  \tilde{J}^{\lambda}_{\bm{x},\infty}(\pi, \gamma;h).
\end{equation}
Moreover, the optimal value of this problem is equal to $\rho$.

\item The stationary policy pair $(\pi_{ss}^\star, \gamma_{ss}^\star)$ is optimal to 
\begin{equation}\label{eq:acopt}
	\min_{\pi \in \bar{\Pi}}\max_{\gamma \in \bar{\Gamma}} \; J^{\lambda}_{\bm{x},\infty}(\pi, \gamma)
\end{equation}
for any  policy spaces $\bar{\Pi} \subset \Pi$ and $\bar{\Gamma} \subset \Gamma$ satisfying
\begin{subequations}
\begin{align}
&\limsup_{T \rightarrow \infty} \frac{1}{T} \mathbb{E}^{\pi_{ss}^\star, \gamma} [h(x_T) \mid  x_0 = \bm{x}] = 0 \quad \forall \gamma \in \bar{\Gamma}\label{eq:avg}\\
&\limsup_{T \rightarrow \infty} \frac{1}{T} \mathbb{E}^{\pi, \gamma_{ss}^\star} [h(x_T) \mid  x_0 = \bm{x}] = 0 \quad \forall \pi \in  \bar{\Pi}. \label{eq:avg2}
\end{align}
\end{subequations}
Moreover, the optimal value of this problem is equal to $\rho$.
\end{enumerate}
\end{theorem}

Theorem~\ref{thm:inf} guarantees the optimality of $(\pi_{ss}^\star, \gamma_{ss}^\star)$ in  the average cost case  under the conditions \eqref{eq:avg} and \eqref{eq:avg2}.
Note that if the mean-state is bounded under $(\pi_{ss}^\star, \gamma)$ so that $\limsup_{T \rightarrow \infty} \lVert \mathbb{E}^{\pi_{ss}^\star, \gamma}[x_T] \rVert< \infty$, then the condition \eqref{eq:avg} is automatically satisfied.
In the following subsection, we will show that $\pi_{ss}^\star$ is BIBO stable, and thereby any $\gamma$ generating a bounded-mean distribution should be contained in $\bar{\Gamma}$.

The condition \eqref{eq:avg2} is a generalization of the one required in the standard average cost LQG to guarantee the optimality of such  a steady-state  policy. If $h(x_T)$ is bounded uniformly over all stages under some policy $\pi$, it clearly satisfies the condition.

\subsubsection{Closed-Loop Stability}

We now discuss the stability properties of the closed-loop system
\begin{equation}\label{cl}
x_{t+1}^\star = (A + B K_{ss})x_t^\star + \Xi w_t + B L_{ss}
\end{equation}
when the optimal policy $\pi_{ss}^\star$ is employed. 
Our first result concerns the expected value of the closed-loop system state $\mathbb{E}[x_t^\star]$, which evolves according to
\begin{equation}\label{ms}
\mathbb{E}[x_{t+1}^\star] = (A + B K_{ss})\mathbb{E}[x_t^\star] + \Xi \mathbb{E}[w_t] + BL_{ss}.
\end{equation}

\begin{theorem}\label{thm:stable1}
Suppose that Assumptions~\ref{ass:pen}--\ref{ass:ob} hold.
Under the optimal policy pair  $(\pi_{ss}^\star, \gamma_{ss}^\star)$,
the expected state of the closed-loop system~\eqref{cl} converges to the following value:
\begin{equation*}
\begin{split}
	&[I - (I + \Phi P_{ss})^{-1}A ]^{-1} [ I-\Phi (I+P_{ss} \Phi - A^\top)^{-1} P_{ss} ]\Xi \bar{w}.
\end{split}
\end{equation*}
Thus, if in addition $\bar{w} = 0$, then the system is linear and
 $\pi_{ss}^\star$ stabilizes the expected state under $\gamma_{ss}^\star$. 
\end{theorem}
 
We can further show that $\pi_{ss}^\star$ guarantees the bounded-input,  bounded-state (BIBO) stability when viewing the disturbance as input.

\begin{theorem}\label{thm:stable2}
Suppose that Assumptions~\ref{ass:pen}--\ref{ass:ob} hold.
Then, the closed-loop gain $(A+BK_{ss})$ is a stable matrix.
Therefore, the mean-state system \eqref{ms} with $\pi_{ss}^\star$ is BIBO stable. 
\end{theorem}

\section{Distributional Robustness with Wasserstein Ambiguity Sets}\label{sec:dr}

\subsection{Finite-Horizon Case}

In the previous section, the minimax control problem with a Wasserstein penalty has been studied in both finite and infinite-horizon settings. 
These results can be used to design a guaranteed-cost controller in the distributionally robust control setting with Wasserstein ambiguity sets~\eqref{eq:const}, as previewed in Section~\ref{sec:dual}. 
We first show  that the total cost under the worst-case distribution in the ambiguity set is bounded as follows:

\begin{lemma}\label{lem:bd} 
For any $\pi \in \Pi$, we have
\begin{equation}\label{eq:bd}
    \sup_{\gamma \in \Gamma_{\mathcal{D}} } J_{\bm{x}}(\pi, \gamma) \leq  \inf_{\lambda \geq 0}  \sup_{\gamma \in \Gamma} \big (  \lambda \theta^2 +   J^\lambda_{\bm{x}} (\pi, \gamma)\big ) \quad \forall \bm{x} \in \mathbb{R}^n.
\end{equation}
\end{lemma}

It is nontrivial to compute the upper-bound for arbitrary $\pi$. 
However, if the optimal policy $\pi^\star$ of our minimax control problem~\eqref{opt} is employed, 
this bound has a tractable form, which is evaluated using the optimal value function of~\eqref{opt}.

\begin{theorem}\label{thm:bd}
Let $\pi^{\star, \lambda}$ be the optimal policy of \eqref{opt} with $\lambda \geq 0$. 
Then, the  cost incurred by $\pi^{\star, \lambda}$ under the worst-case distribution policy in $\Gamma_{\mathcal{D}}$ is bounded as follows:
\[
\sup_{\gamma \in \Gamma_\mathcal{D}} J_{\bm{x}} (\pi^{\star, \lambda}, \gamma) 
    \leq \lambda \theta^2 +    V (\bm{x}; \lambda) \quad \forall \bm{x} \in \mathbb{R}^n.
    \]
\end{theorem}

The dependence of the upper-bound on penalty parameter $\lambda$ 
indicates that the distributional robustness of our policy $\pi^{\star, \lambda}$ can be controlled by tuning $\lambda$.
This theorem can be used to select an optimal penalty parameter $\lambda^\star$ that provides the least upper-bound. 
Given $\bm{x} \in \mathbb{R}^n$,
let $\lambda^\star$ be defined by
\begin{equation}\label{opt_lambda}
\lambda^\star \in \argmin_{\lambda \geq 0}  \big (\lambda\theta^2 +V (\bm{x}; \lambda)  \big ).
\end{equation}
Then, the  cost incurred by $\pi^{\star, \lambda^\star}$ under the worst-case distributions in the ambiguity sets is bounded as follows:
\[
\sup_{\gamma \in \Gamma_\mathcal{D}} J_{\bm{x}} (\pi^{\star, \lambda^\star}, \gamma) 
    \leq \lambda^\star \theta^2 +    V (\bm{x}; \lambda^\star).
    \]
To solve the minimization problem in \eqref{opt_lambda}, 
we first identify some structural properties of $V (\bm{x}; \lambda)$ using the results in Section~\ref{sec:finite}.

\begin{lemma}\label{lem:lam}
Let $P_t^\lambda$, $r_t^\lambda$  and $z_t^\lambda$, $t = 0, \ldots, T$, be obtained by the Riccati equation~\eqref{ric} with $P_T^\lambda=Q_f$, $r_T^\lambda=0$, and $z_T^\lambda=0$,  given $\lambda \geq 0$.
Let 
\begin{equation} \nonumber
\hat{\lambda} := \inf \{\lambda \mid  \lambda I - \Xi^\top P_t^\lambda \Xi \succ 0, \; t=1,2, \ldots, T\}.
\end{equation}
Then,
\begin{equation*}
V(\bm{x}; \lambda)  = 
\begin{cases}
\infty & \text{if} \; \; \lambda \in [0, \hat{\lambda})\\ 
c_1 & \text{if} \; \; \lambda = \hat{\lambda}\\
c_2(\lambda) & \text{if} \; \; \lambda \in (\hat{\lambda}, \infty),
\end{cases}
\end{equation*}
where
$c_2(\lambda) := (\bm{x}^\top P_0^\lambda \bm{x} + 2(r_0^\lambda)^\top \bm{x} + z_0^\lambda)/T$ and $c_1$ is a constant satisfying the boundary condition $c_1 \geq c_2(\hat{\lambda}+\epsilon)$ for all $\epsilon >0$.
\end{lemma}

This structural property of the optimal value function yields the following simple way to find a minimizer $\lambda^\star$ of \eqref{opt_lambda}.

\begin{proposition}\label{prop:lam}
Suppose that 
\begin{equation}\label{eq:min}
\lambda_* \in    \argmin_{\lambda > \hat{\lambda}} \bigg [ \lambda \theta^2 + \frac{1}{T} \big (\bm{x}^\top P_0^\lambda \bm{x} + 2(r_0^\lambda)^\top \bm{x} + z_0^\lambda \big ) \bigg ].
\end{equation}
Then, $\lambda_*$ is a minimizer of the optimization problem \eqref{opt_lambda}, i.e., $\lambda_* = \lambda^\star$.
Moreover, the optimization problem~\eqref{eq:min} is convex. 
\end{proposition}

As shown in the proof of Lemma~\ref{lem:lam}, $\hat{\lambda}$ is the unique boundary point that separates the range of  $\lambda$ by whether it satisfies Assumption~\ref{ass:pen} or not.
Therefore, $\hat{\lambda}$ can be obtained by binary search.
An optimal $\lambda^\star$  can then be computed by solving \eqref{eq:min} with existing convex optimization algorithms.

\subsection{Infinite-Horizon Case}
In this subsection, we examine the distributional robustness of the steady-state optimal  policy using the following average cost criterion:
\begin{equation}\label{av_cost}
\begin{split}
&J_{\bm{x}, \infty}(\pi, \gamma) :=  \limsup_{T \rightarrow \infty} \frac{1}{T} \mathbb{E}^{\pi, \gamma} \bigg[ \sum_{t=0}^{T-1} ( x_t^\top Q x_t 
+ u_t^\top R u_t )~\bigg\vert~x_0 = \bm{x} \bigg].
\end{split}
\end{equation}

Fix any penalty parameter $\lambda > 0$ satisfying Assumptions~{\ref{ass:pen}--\ref{ass:ob}}.
Let $P_{ss}^\lambda$ be the unique symmetric PSD solution of the ARE~\eqref{are} and $r_{ss}^\lambda$ be defined as \eqref{eq:oss} with given $\lambda \geq 0$.
The corresponding stationary cost in Proposition~\ref{prop:avgcost} is denoted by $\rho(\lambda)$.
As in Section~\ref{sec:infinite}, we use an extended average cost function including the bias $h^\lambda (\bm{x}) := \bm{x}^\top P_{ss}^\lambda \bm{x} + 2 (r_{ss}^\lambda)^\top \bm{x}$, defined as
\begin{equation}\nonumber
\begin{split}
&\tilde{J}_{\bm{x}, \infty}(\pi, \gamma; h^\lambda) :=  \limsup_{T \rightarrow \infty} \frac{1}{T} \mathbb{E}^{\pi, \gamma} \bigg[ \sum_{t=0}^{T-1} ( x_t^\top Q x_t 
+ u_t^\top R u_t ) + h^\lambda (x_T)~\bigg\vert~x_0 = \bm{x} \bigg].
\end{split}
\end{equation}

Consider $\pi^{\star, \lambda}_{ss}$ 
 obtained in Corollary~\ref{cor:inf} with  the  penalty parameter $\lambda$.
The average cost incurred by this policy under the worst-case distribution in the Wasserstein ambiguity set 
\[
\mathcal{D} := \{ \mu \in \mathcal{P}(\mathbb{R}^k) \mid W_2 (\mu, \nu) \leq \theta \}
\]
is computed as
\[
\sup_{\gamma \in \Gamma_{\mathcal{D}}} J_{\bm{x}, \infty}(\pi_{ss}^{\star, \lambda}, \gamma).
\]
The worst-case cost is uniformly bounded by a constant depending on $\lambda$.

\begin{theorem} \label{thm:gc}
Suppose that Assumptions~\ref{ass:pen}--\ref{ass:ob} hold for some penalty parameter $\lambda$.
If the stationary policy $\pi^{\star, \lambda}_{ss}$, defined in Corollary~\ref{cor:inf}, is employed, 
then the worst-case extended average cost is bounded as follows:
\begin{equation}\label{gc1}
    \sup_{\gamma \in \Gamma_{\mathcal{D}}}  \tilde{J}_{\bm{x}, \infty}(\pi^{\star, \lambda}_{ss}, \gamma; h^\lambda)   \leq {\lambda} \theta^2 + \rho({\lambda}) \quad \forall \bm{x} \in \mathbb{R}^n.
\end{equation}
Moreover, for any policy space $\bar{\Gamma}_{\mathcal{D}} \subset \Gamma_{\mathcal{D}}$ satisfying
\[
\limsup_{T \to \infty} \frac{1}{T} \mathbb{E}^{\pi_{ss}^{\star, \lambda}, \gamma}[ h^\lambda (x_T) \mid x_0 = \bm{x} ] = 0 \quad \forall \gamma \in \bar{\Gamma}_{\mathcal{D}},
\]
the worst-case  average cost has the same upper-bound:
\begin{equation}\label{gc2}
    \sup_{\gamma \in \bar{\Gamma}_{\mathcal{D}}}  J_{\bm{x}, \infty}(\pi^{\star, \lambda}_{ss}, \gamma)   \leq {\lambda} \theta^2 + \rho({\lambda}) \quad \forall \bm{x} \in \mathbb{R}^n.
\end{equation}
\end{theorem}

This theorem indicates the robustness of our policy $\pi_{ss}^{\star, \lambda}$ against any distribution errors within the Wasserstein ball.
Since the upper-bound depends on  the penalty parameter $\lambda$, 
it is important to select an appropriate value of $\lambda$. 
Here, we present a way to obtain a suboptimal $\lambda$ that minimizes the upper-bound over a certain range of $\lambda$.

\begin{lemma}\label{lem:gc}
Suppose that Assumptions~\ref{ass:st} and \ref{ass:ob} hold, and $(A, B)$ is stabilizable.
Then, there exists $\hat{\lambda}_1>0$ such that Assumption \ref{ass:W} holds for any $\lambda > \hat{\lambda}_1$.
\end{lemma}

\begin{proposition}\label{prop:gc}
Suppose that Assumptions~\ref{ass:st} and \ref{ass:ob} hold, and $(A, B)$ is stabilizable.
Let  $\hat{\lambda}_1$ be the constant defined in Lemma~\ref{lem:gc} and assume that
\begin{equation}\label{eq:pen}
\hat{\lambda}_2 := \inf \{\lambda \mid  \lambda I - \Xi^\top P_t^\lambda \Xi \succ 0 \;\; \forall t \geq 1, \forall T\geq 1\}
\end{equation}
is finite.
Let $\hat{\lambda}_\infty:= \max_{i=1,2}\{\hat{\lambda}_{i}\}$.
Then, $\rho(\lambda)$ defined in Proposition~\ref{prop:avgcost}  is a monotonically nonincreasing convex function  on $(\hat{\lambda}_\infty, \infty)$.
Moreover,
\begin{equation}\label{lqgavgcost}
\begin{split}
 \lim_{\lambda \rightarrow \infty} \rho(\lambda) &= \mathrm{tr}[ \Xi^\top \tilde{P}_{ss} \Xi \Sigma] - \bar{w}^\top \Xi^\top   \tilde{P}_{ss} B R^{-1}B^\top (I +  \tilde{P}_{ss} B R^{-1}B^\top )^{-1}  \tilde{P}_{ss} \Xi \bar{w}\\
& + (2  \bar{w}^\top \Xi^\top - \tilde{r}_{ss}^\top  B R^{-1}B^\top)(I+ \tilde{P}_{ss} B R^{-1}B^\top)^{-1}  \tilde{r}_{ss}.
\end{split}
\end{equation}
Here,
$\tilde{P}_{ss}:=\lim_{T \rightarrow \infty}  \tilde{P}_t$ and $ \tilde{r}_{ss}:=\lim_{T \rightarrow \infty}  \tilde{r}_t$, where $\tilde{P}_t$ and $\tilde{r}_t$ are generated by the Riccati equation \eqref{ric_lqg} for the standard LQG with $\tilde{P}_T =Q_f$, $\tilde{r}_T = 0$, and $\tilde{z}_T =0$.
\end{proposition}

Under the assumptions in Proposition~\ref{prop:gc}, one may 
consider the convex optimization problem
\begin{equation}\label{conv}
\min_{\lambda > \hat{\lambda}_\infty} [\lambda\theta^2 + \rho(\lambda)]
\end{equation}
to  find a reasonably tight upper-bound of the average-cost.
The values of $\hat{\lambda}_1$ in Lemma~\ref{lem:gc} and $\hat{\lambda}_2$ in Proposition~\ref{prop:gc} are required for computing $\hat{\lambda}_\infty$.
The first parameter $\hat{\lambda}_1$ can be obtained by examining the eigenvalues of $\Phi$ and $A+ \Phi^{1/2} R^{1/2}K$.
Specifically, reducing $\lambda$ from a sufficiently large value, one can compute the biggest value of $\lambda$ such that $\Phi \succeq 0$ and 
all eigenvalues of $A+ \Phi^{1/2} R^{1/2}K$  lie  inside the unit circle. 
The second parameter $\hat{\lambda}_2$ can be considered as the infinite-horizon extension of $\hat{\lambda}$ in Lemma~\ref{lem:lam}.
As in the finite-horizon case,  one can obtain $\hat{\lambda}_2$ via  binary search.
Note that the optimization problem~\eqref{conv} provides the least upper-bound over the range of $\lambda$ that  the existence of $\rho(\lambda)$ is guaranteed.
Nevertheless, when $\hat{\lambda}_\infty = \hat{\lambda}_2 \geq \hat{\lambda}_1$, \eqref{conv} provides the least bound for a nearly entire range, $\lambda \in [0, \hat{\lambda}_2)\cup (\hat{\lambda}_2, \infty)$, since the optimal value is $+\infty$ for $\lambda < \hat{\lambda}_2$.

\subsection{Out-of-Sample Performance Guarantee}

An advantage of using the Wasserstein metric in distributionally robust control is to attain a performance guarantee measured under a new sample, independent of $\hat{w} = (\hat{w}^{(1)}, \ldots, \hat{w}^{(N)})$ used in the controller design. Such an out-of-sample performance guarantee has been studied in the infinite-horizon discounted cost setting~\cite{Yang2020}.
In this section, we extend this result to the finite-horizon and the infinite-horizon average cost cases. Throughout this subsection, we fix $\lambda >0$.

\subsubsection{Finite-Horizon Case}
Let $(\pi^\star_{\hat{w}}, \gamma^\star_{\hat{w}})$ denote the optimal policy pair of the finite-horizon minimax control problem~\eqref{opt} with sample $\hat{w} = (\hat{w}^{(1)}, \ldots, \hat{w}^{(N)})$  and the penalty parameter $\lambda$.
The out-of-sample performance of $\pi^\star_{\hat{w}}$ is defined as
\begin{equation}
\frac{1}{T}\mathbb{E}_{w \sim \mu}^{\pi^\star_{\hat{w}}} \bigg [
\sum_{t=0}^{T-1} (x_t^\top Q x_t + u_t^\top R u_t) + x_T^\top Q_f x_T~\bigg\vert~x_0 = \bm{x}
\bigg],
\end{equation}
where $\mu := \mu_0 \times \cdots \times \mu_{T-1} \in \mathcal{P}(\mathbb{R}^{k\times T})$ represents the true but unknown distribution of $w$.
Note that it is intractable to explicitly evaluate the out-of-sample performance since $\mu$ is unknown. 
As an alternative, the following probabilistic guarantee can be considered:
\begin{equation}\label{eq:out}
\begin{split}
\mu^N \bigg\{  \hat{w} :  \frac{1}{T} &\mathbb{E}^{\pi^\star_{\hat{w}}}_{w \sim \mu} \big[C_T(x, u) \mid x_0 = \bm{x} \big]    \leq \lambda \theta^2 +  V (\bm{x}; \lambda) \quad \forall \bm{x} \in \mathbb{R}^n \bigg \} \geq 1-\beta,
\end{split}
\end{equation}
where 
\[
C_T(x, u) := x_T^\top Q_f x_T + \sum_{t=0}^{T-1} (x_t^\top Q x_t + u_t^T R u_t)
\]
denotes the total cost,
and
$\beta \in (0, 1)$ represents an acceptable error in probability.
Our goal is to identify a condition on the size of the Wasserstein ambiguity set for satisfying the probabilistic out-of-sample performance guarantee. 
To begin, we assume that $\mu$ has a light tail. 
\begin{assumption}\label{ass:tail}
There exist $p > 0$ and $q > 2$ satisfying
\begin{equation*}
   \int_{\mathbb{R}^k} \exp({p \|w\|^q}) \mathrm{d} \mu_t(w) < \infty  
\end{equation*}
for $t = 0, \ldots, T-1$.
\end{assumption} 
Under this assumption, the following measure concentration inequality holds for the Wasserstein metric \cite{Fournier2015}[Theorem 2].
\begin{lemma}\label{lem:out}
Suppose that Assumption~\ref{ass:tail} holds.
Let $\nu_{\hat{w}_t}$ denote the empirical distribution~\eqref{emp} obtained using sample $\hat{w}_t$. 
Then, for all $N \geq 1$, $\theta >0$ and $t = 0, \ldots, T-1$, we have
\begin{equation*}
\begin{split}
\mu_t^N \big \{ \hat{w}_t : W_2 &(\mu_t, \nu_{\hat{w}_t})^2 \geq \theta \big  \}  \leq c_1 [b_1(N, \theta) \mathbf{1}_{\{\theta \leq 1\}} + b_2(N, \theta) \mathbf{1}_{\{\theta >1\}}],
\end{split}
\end{equation*}
where
\begin{equation*}
    b_1(N, \theta) := \begin{cases}
\exp(-c_2 N \theta^2) & \text{if} \; \; k<4\\ 
\exp(-c_2 N (\frac{\theta}{log(2+1 / \theta)})^2) & \text{if} \; \; k=4\\
\exp(-c_2 N \theta^{k/2}) & \text{if} \; \; k>4,
\end{cases} 
\end{equation*}
and
\begin{equation*}
    b_2(N, \theta):= \exp(-c_2 N \theta^{q / 2}).
\end{equation*}
The positive constants $c_1$ and $c_2$ depend  only on $k$, $p$ and $q$.
\end{lemma}

This lemma provides a sufficient condition for the probabilistic out-of-sample performance guarantee~\eqref{eq:out}.

\begin{theorem}\label{thm:out}
Suppose that Assumptions~\ref{ass:pen} and~\ref{ass:tail} hold. Let the radius $\theta$ be chosen as
\begin{equation*}
    \theta (N, \beta)  =
    \begin{cases}
    c^{1/q} & \text{if} \; \;c>1 \\
    c^{1/4} & \text{if} \; \;c \leq 1 \wedge k<4 \\
     c^{1/k} & \text{if} \; \;c \leq 1 \wedge k>4\\
    \bar{\theta} & \text{if} \; \;c \leq 1/(\log 3)^2 \wedge k=4,\\
    \end{cases}
\end{equation*}
where
\begin{equation*}
    c:= \frac{1}{N c_2} \log \bigg( \frac{c_1}{1-(1-\beta)^{1/T}} \bigg),
\end{equation*}
$c_1$ and $c_2$ are the positive constants defined in Lemma~\ref{lem:out}, and $\bar{\theta}$ satisfies
\begin{equation*}
    \frac{\bar{\theta}^2}{\log (2+ 1/\bar{\theta}^2)} = {c}^{1/2}.
\end{equation*}
Then, the probabilistic out-of-sample performance guarantee \eqref{eq:out} holds.
\end{theorem}

Under the more strict assumption that each $\mu_t$ is compactly supported rather than having a light tail, a simpler concentration inequality holds as proposed in~\cite{Boskos2020}[Proposition 3.2].
Define the diameter of a set $S \subset \mathbb{R}^{k}$ as $\mathrm{diam}(S):=\sup\{\lVert x-y\rVert_\infty \mid x, y \in S\}$.
Let $\mathrm{supp}(\mu)$ denote the smallest closed set that has measure $1$ with $\mu$.
Then,  using the same argument as that in the proof of Theorem~\ref{thm:out},
we can show that the following guarantee holds.

\begin{cor}\label{cor:comp}
Suppose that  Assumption~\ref{ass:pen} holds and $\mu_t$'s are compactly supported for all $t = 0, \ldots, T-1$ with $\zeta := \frac{1}{2} \max\{ \mathrm{diam}(\mathrm{supp}(\mu_t)) \mid t=0,\ldots,T-1\}$.
Let the radius $\theta$ be chosen as
\begin{equation*}
    \theta(N, \beta):=
    \begin{cases}
    c^{1/4}\zeta & \text{if} \;\;  k<4 \\
     c^{1/k}\zeta & \text{if} \;\; k>4\\
    \bar{\theta} & \text{if} \;\; k=4, \\
    \end{cases}
\end{equation*}
where
\begin{equation*}
    c:= \frac{1}{N c_2} \log \bigg( \frac{c_1}{1-(1-\beta)^{1/T}} \bigg),
\end{equation*}
$c_1$ and $c_2$ are the positive constants depend only on $k$, and $\bar{\theta}$ satisfies
\begin{equation*}
    \frac{\bar{\theta}^2}{\zeta^2 \log (2+ \zeta^2/\bar{\theta}^2)} = {c}^{1/2}.
\end{equation*}
Then, the probabilistic out-of-sample performance guarantee \eqref{eq:out} holds.
\end{cor}

A potential disadvantage of directly employing the radius suggested in Theorem~\ref{thm:out} or Corollary~\ref{cor:comp} is that the guaranteed upper-bound $\lambda \theta(N, \beta)^2 + V (\bm{x}; \lambda)$ grows with the number of stage $T$.
Specifically, $c$ increases logarithmically with $T$, since $1-(1-\beta)^{1/T} \approx \beta/T$ when $T$ is sufficiently large.
The stationarity assumption for $\mu_t$ can be used to alleviate this issue.
For instance, if we assume that $T' \in [1, T]$ stages have a stationary probability distribution and therefore use only one sample for $T'$ stages, then we can reduce $T$ to $T-T'+1$ in our formulation of $\theta(N, \beta)$.
Thus, in the case with Assumption~\ref{ass:st}, we can simply replace $T$ by $1$.

\subsubsection{Infinite-Horizon Case}

We now consider the infinite-horizon case with the average cost criteria \eqref{av_cost}.
Under  Assumptions~\ref{ass:pen}--\ref{ass:ob}, let $\pi^{\star}_{ss,\hat{w}}$ denote the optimal policy obtained in Corollary~\ref{cor:inf} with sample $\hat{w} = (\hat{w}^{(1)}, \ldots, \hat{w}^{(N)})$.
 We assume that the true distribution is stationary, i.e., $\mu_t \equiv \mu$ for all $t$. 
In this average cost setting,  our interest is to study the following probabilistic bound on the out-of-sample performance of $\pi^{\star}_{ss,\hat{w}}$:
\begin{equation}\label{eq:out2}
\begin{split}
\mu^N \bigg \{  \hat{w} : \limsup_{T \rightarrow \infty} \frac{1}{T} \mathbb{E}^{\pi^{\star}_{ss,\hat{w}}}_{w \sim \mu} \bigg[ \sum_{t=0}^{T-1} ( x_t^\top Q x_t 
+ u_t^\top R u_t ) \bigg\vert x_0 = \bm{x} \bigg]
    \leq  \lambda \theta^2 + \rho(\lambda) \; \forall \bm{x} \in \mathbb{R}^n \bigg \} \geq 1-\beta,
\end{split}
\end{equation}
where $\lambda \theta^2 + \rho(\lambda)$ is the upper-bound on the average cost in Theorem~\ref{thm:gc}.
\begin{cor}\label{cor:out}
Suppose that Assumptions~\ref{ass:pen}--\ref{ass:tail} hold. Let the radius $\theta$  be chosen as
\begin{equation*}
    \theta(N, \beta):=
    \begin{cases}
    c^{1/q} & \text{if} \; \;c>1 \\
    c^{1/4} & \text{if} \; \;c \leq 1 \wedge k<4 \\
     c^{1/k} & \text{if} \; \;c \leq 1 \wedge k>4\\
    \bar{\theta} & \text{if} \; \;c \leq 1/(\log 3)^2 \wedge k=4,\\
    \end{cases}
\end{equation*}
where
\begin{equation*}
    c:= \frac{1}{N c_2} \log \bigg( \frac{c_1}{\beta} \bigg),
\end{equation*}
$c_1$ and $c_2$ are the positive constants defined in Lemma~\ref{lem:out}, and $\bar{\theta}$ satisfies
\begin{equation*}
    \frac{\bar{\theta}^2}{\log (2+ 1/\bar{\theta}^2)} = c^{1/2}.
\end{equation*}
Then, the probabilistic out-of-sample performance guarantee \eqref{eq:out2} holds.
\end{cor}

If $\mu$ is contained in $\mathcal{D}$, then the policy $\gamma_t \equiv \mu$ must be contained in $\bar{\Gamma}_{\mathcal{D}}$, implying that the guaranteed-cost property~\eqref{gc2} holds.
Since the rest of proof is similar to that for Theorem~\ref{thm:out}, we have omitted the proof. 
Note that the radius in Corollary~\ref{cor:out} can be obtained by letting $T=1$ in Theorem~\ref{thm:out} due to the  stationarity assumption of $\mu$ and $\nu$.
The case of compactly supported distributions can be considered similarly to Corollary~\ref{cor:comp}.

\section{Relations to $H_\infty$-Optimal Control}\label{sec:Hinfinity}
In this section, we discuss relations between our minimax control method and the $H_\infty$-method.
For comparison,
we consider the classical dynamic game formulation for minimizing the $H_\infty$-norm of the cost function with respect to the disturbance~(e.g., \cite{Basar2008}). 

\subsection{Finite-Horizon Case}
We first examine the finite-horizon case with the initial condition $x_0 = 0$. 
For $H_\infty$-control, we consider a modified dynamic game problem, where the opponent's policy $\tilde{\gamma}_t$ now maps the current state $x_t$ to disturbance vector $w_t$ rather than its distribution~\cite{Basar2008}. Note that the disturbance vector is no longer random in the $H_\infty$-setting. 
The set of admissible opponent's policies is accordingly modified and is denoted by~$\tilde{\Gamma}$.
Consider the following quadratic cost function:
\[
\tilde{J} (\pi, \tilde{\gamma}) := \mathbb{E}^{\pi, \tilde{\gamma}} \bigg[ \sum_{t=0}^{T-1} ( x_t^\top Q x_t + u_t^\top R u_t ) + x_T^\top Q_f x_T~\bigg\vert~x_0 = 0 \bigg].
\] 
Given a control policy $\pi$,
we seek to find the infimum of $\lambda > 0$ such that 
\begin{equation*}
   \sup_{\tilde{\gamma} \in \tilde{\Gamma} : \lVert w \rVert \leq 1} \tilde{J} (\pi, \tilde{\gamma}) = \sup_{\tilde{\gamma} \in \tilde{\Gamma} } \frac{\tilde{J}(\pi, \tilde{\gamma})}{\| w \|^2 } \leq \lambda,
\end{equation*} 
where $\| w \|^2 := \sum_{t=0}^{T-1} \lVert w_t \rVert^2$. 
The first equality holds since $\tilde{J} (\pi, \tilde{\gamma})$ is homogeneous with respect to $\| w\|^2$ when $x_0 = 0$. 
Note also that ${\tilde{J}(\pi, \tilde{\gamma})}/{\sum_{t=0}^{T-1} \lVert w_t \rVert^2} \leq \lambda$ for all $\tilde{\gamma} \in \tilde{\Gamma}$ if and only if $\tilde{J}(\pi, \tilde{\gamma}) - \lambda \sum_{t=0}^{T-1} \lVert w_t \rVert^2 \leq 0$ for all $\tilde{\gamma} \in \tilde{\Gamma}$.
Thus, the inequality above can be rewritten as
\[
 \sup_{\tilde{\gamma} \in \tilde{\Gamma} } \bigg [
 \tilde{J}(\pi, \tilde{\gamma}) - \lambda \sum_{t=0}^{T-1} \| w_t \|^2
 \bigg ]\leq 0.
\]
This motivates us to consider the following augmented cost function with an additional disturbance-norm term:
\begin{equation*}
\begin{split}
\tilde{J}^{\lambda} (\pi, \tilde{\gamma}) := \mathbb{E}^{\pi, \tilde{\gamma}} \bigg[ \sum_{t=0}^{T-1} 
( x_t^\top Q x_t + u_t^\top R u_t - \lambda \| w_t \|^2 )
 + x_T^\top Q_f x_T~\bigg\vert~x_0 = 0 \bigg],
\end{split}
\end{equation*}
as well as the following minimax control problem:
\[
\tilde{J}^{\lambda, \star} := \inf_{\pi \in \Pi} \sup_{\tilde{\gamma} \in \tilde{\Gamma}} \tilde{J}^{\lambda} (\pi, \tilde{\gamma}).
\]
Let $\Lambda:= \{ \lambda \mid \tilde{J}^{\lambda, \star} \leq 0\}$.
The desired $\lambda^\star$ can be obtained as 
$\lambda^\star := \inf \{ \lambda \mid \lambda \in \Lambda \}$.
More details about the dynamic game formulation of $H_\infty$-control can be found in~\cite[Section 1.4]{Basar2008}.
Let $\tilde{V}_t:\mathbb{R}^n \to \mathbb{R}$ denote the value function of this problem. The dynamic programming principle gives the following Bellman equation:
\begin{equation}\nonumber
\begin{split}
&\tilde{V}_t (\bm{x}) = \bm{x}^\top Q \bm{x} + \inf_{\bm{u} \in \mathbb{R}^m} \bigg [
\bm{u}^\top R \bm{u} + \sup_{\bm{w} \in \mathbb{R}^k} \{
\tilde{V}_{t+1} (A\bm{x} + B\bm{u} + \Xi \bm{w}) - \lambda \| \bm{w} \|^2
\}
\bigg ]
\end{split}
\end{equation}
with $\tilde{V}_T (\bm{x}) := \bm{x}^\top Q_f \bm{x}$.
If we parameterize $\tilde{V}_t (\bm{x}) = \bm{x}^\top P_t \bm{x}$ under Assumption~\ref{ass:pen}, the Riccati equation and the control gain $K_t$ are obtained as those in Section~\ref{sec:finite}.
Note that there are no  $r_t$, $z_t$, and $L_t$ terms in the $H_\infty$-control. 
The worst-case disturbance policy is then given by
\[
\tilde{\gamma}_t^\star (\bm{x}) = (\lambda I - \Xi^\top P_{t+1} \Xi)^{-1} \Xi^\top P_{t+1} (A+ BK_t) \bm{x}.
\]
Since $x_0 = 0$, we deduce that
\[
\tilde{J}^{\lambda, \star} = \tilde{V}_0 (0) = 0^\top P_0 0 = 0.
\]
In this case, any $\lambda$ satisfying Assumption~\ref{ass:pen} must be contained in $\Lambda$. 
However, if $\lambda$ does not satisfy Assumption~\ref{ass:pen}, then the cost value would be $+\infty$, and thus $\lambda$ cannot belong to $\Lambda$.
Thus, we conclude that $\lambda^\star$ is the infimum of $\lambda$ that satisfies Assumption~\ref{ass:pen}. See \cite{Basar2008}[Section 3.3] for further details on the optimal disturbance attenuation level $\lambda^\star$ of $H_\infty$-control.

The worst-case disturbance $\tilde{\gamma}_t^\star (\bm{x})$ in the $H_\infty$-method is related with the support elements  $w_t^{\star, (i)} (\bm{x})$ of the worst-case distribution in Corollary~\ref{cor:dist} in our method as follows:
\begin{equation*}
\begin{split}
w_t^{\star, (i)} (\bm{x}) = \tilde{\gamma}_t^\star (\bm{x})  &+ (\lambda I - \Xi^\top P_{t+1} \Xi)^{-1}  (\Xi^\top P_{t+1} BL_t + \Xi^\top r_{t+1} + \lambda \hat{w}^{(i)}_t ).
\end{split}
\end{equation*}
This indicates that each support element of the worst-case distribution in Corollary~\ref{cor:dist} can be considered to be shifted from $\tilde{\gamma}_t^\star (\bm{x})$ by the scaled terms generated from the sample data $\hat{w}^{(i)}_t$ and  $L_t$, $r_{t+1}$.\footnote{If the sample mean is  zero, i.e., $\sum_{i=1}^N \hat{w}^{(i)}_t =0$, for all stages, then $\tilde{\gamma}_t^\star (\bm{x})$ corresponds to the mean value of the worst-case distribution.}
Thus, our minimax control method with Wasserstein distance can be understood as a distributional generalization of the $H_\infty$-method.

\subsection{Infinite-Horizon Case}

In the infinite-horizon case, the corresponding $H_\infty$-control can be obtained using a limiting solution of the Riccati equation. 
This yields the same ARE as \eqref{are} for our minimax control methods~\cite[Section 3.4]{Basar2008}. 
Under Assumptions~\ref{ass:pen}--\ref{ass:ob}, the ARE has a symmetric PSD solution $P_{ss}$ from which we obtain the same control gain  $K_{ss}$.
Regarding the worst-case disturbance, we have 
\[
\tilde{\gamma}^\star (\bm{x}) = (\lambda I - \Xi^\top P_{ss} \Xi)^{-1} \Xi^\top P_{ss} (A+ BK_{ss}) \bm{x}.
\]
Thus, 
 the worst-case disturbance  in the $H_\infty$ method is related to our  worst-case distribution $\gamma^\star_{ss} (\bm{x}) := \frac{1}{N} \sum_{i=1}^N \delta_{w^{\star, (i)} (\bm{x})}$ through  
\begin{equation*}
\begin{split}
w^{\star, (i)} (\bm{x}) = \tilde{\gamma}^\star (\bm{x})  &+  (\lambda I - \Xi^\top P_{ss} \Xi)^{-1}  (\Xi^\top P_{ss} BL_{ss}  + \Xi^\top r_{ss} + \lambda \hat{w}^{(i)}).
\end{split}
\end{equation*}

The relationship between our method and the $H_\infty$-method enables us to analyze our controller using the classical robust stability results.
Consider a dynamical system of the form
\begin{equation*}
\begin{split}
&x_{t+1} = Ax_t + Bu_t + \Xi w_t, \\
&z_t = \begin{bmatrix} {Q}^{1/2}\; \\ 0 \end{bmatrix} x_t + \begin{bmatrix} 0 \\ {R}^{1/2}\; \end{bmatrix} u_t, \\
\end{split}
\end{equation*}
where $z_t$ is the  error output and $w_t$ is the disturbance input.
Let $T_\pi$ denote the closed loop transfer function from input $w$ to output $z$ under control policy $\pi$.
As mentioned above, our minimax controller is equivalent to the $H_\infty$-controller when $x_0 = 0$ and the empirical distribution has zero mean, i.e., $\bar{w} = 0$. 
This equivalence yields the following robust stability property of our controller (see e.g., \cite{Basar2008} and \cite{Glover1988} for further details about robust stability).
\begin{proposition}
Suppose that Assumptions~\ref{ass:pen}--\ref{ass:ob} hold, $x_0 = 0$,  and $\bar{w}=0$.
Then, the optimal minimax control policy $\pi^\star_{ss}$ in Corollary~\ref{cor:inf} satisfies the following robust stability property:
\[
\| T_{\pi^\star_{ss}} \|_\infty := \sup_{w \in \mathbb{R}} \bar{\sigma} (T_{\pi^\star_{ss}}(jw)) \leq \lambda,
\]
where $\bar{\sigma} (T(jw))$ denotes the largest singular value of $T(jw)$.
\end{proposition}

Conversely, our stochastic interpretation of the classical $H_\infty$-method enables us to analyze the $H_\infty$-controller 
 from the distributional perspective,  particularly when the sample of $w_t$ is available. 
For instance, in such data-driven scenarios, one can obtain the probabilistic performance guarantee of the $H_\infty$-controller using the out-of-sample performance result in the previous section. 
To be more precise, for the $H_\infty$-controller with a fixed $\lambda$,
its out-of-sample performance satisfies the probabilistic bound~\eqref{eq:out2}.

\section{Numerical Experiments}
In this section, the performance of our minimax control method is demonstrated through  a power system frequency regulation problem.
Stability is an important  issue in power transmission systems, as the penetration of variable renewable energy sources 
and the potential of data integrity attacks
 increase.
We apply the minimax control method on the IEEE 39 bus system, which models the New England power grid and has been frequently used to evaluate frequency control methods (e.g. \cite{Dorfler2014, Dizche2019}).
This model consists of 39 buses, 46 lines, and 10 generators.
We use a classical generator model without an excitation system, such as a power system stabilizer  and an automatic voltage regulator, for simplicity.

Let $\delta_i$ and  $\omega_i$ denote the rotor angle and the frequency of the $i$th generator. Then, $\dot{\delta_i} = \omega_i - \omega_s$, where $\omega_s$ is a constant synchronous speed. 
The electromechanical swing equation for  the $i$th generator is given by the following damped oscillator:
\begin{equation*}
    \frac{2H_i}{\omega_s} \dot{\omega_i} =  P_{i} - d_i \omega_i - \sum_{j \neq i} \vert Y_{ij} \vert E_i E_j \sin(\delta_i - \delta_j),
\end{equation*}
where
$H_i$, $P_i$, $d_i$, and $E_i$ denote the inertia, the power injection, the damping coefficient, and the voltage of the $i$th generators, and
$Y$ denotes the admittance matrix of the power network.
Linearizing the swing equations at an operating point $(\delta^*, \omega^*)$ yields
\begin{equation*}
    M \Delta \ddot{\delta} + D \Delta\dot{\delta} + L \Delta\delta = \Delta P,
\end{equation*}
where
$M := (2/\omega_s)\mathrm{diag}(H)$, $D := \mathrm{diag}(d)$, and the Kron-reduced Laplacian matrix $L$ is defined by $L_{ij} := - \vert Y_{ij} \vert E_i E_j \cos(\delta^*_i-\delta^*_j)$ for $i \neq j$ and $L_{ii} := -\sum_{j \neq i} L_{ij}$.
The second-order ordinary differential equation can be expressed in the following state-space form:
\begin{equation*}
    \begin{bmatrix}\Delta \dot{\delta} \\\Delta \dot{\omega} \end{bmatrix} = \underbrace{\begin{bmatrix} 0 & I \\ -M^{-1} L & -M^{-1} D \end{bmatrix}}_{=: A} \begin{bmatrix} \Delta \delta \\ \Delta \omega \end{bmatrix}+ \underbrace{\begin{bmatrix} 0 \\ M^{-1} \end{bmatrix}}_{=:B}\Delta P,
\end{equation*}
with system state $x(t):=(\Delta\delta^\top(t), \Delta\omega^\top(t))^\top \in \mathbb{R}^{20}$ and control input $u(t):=\Delta P(t) \in \mathbb{R}^{10}$.

\begin{figure}[t!]
\centering
\includegraphics[scale=0.8]{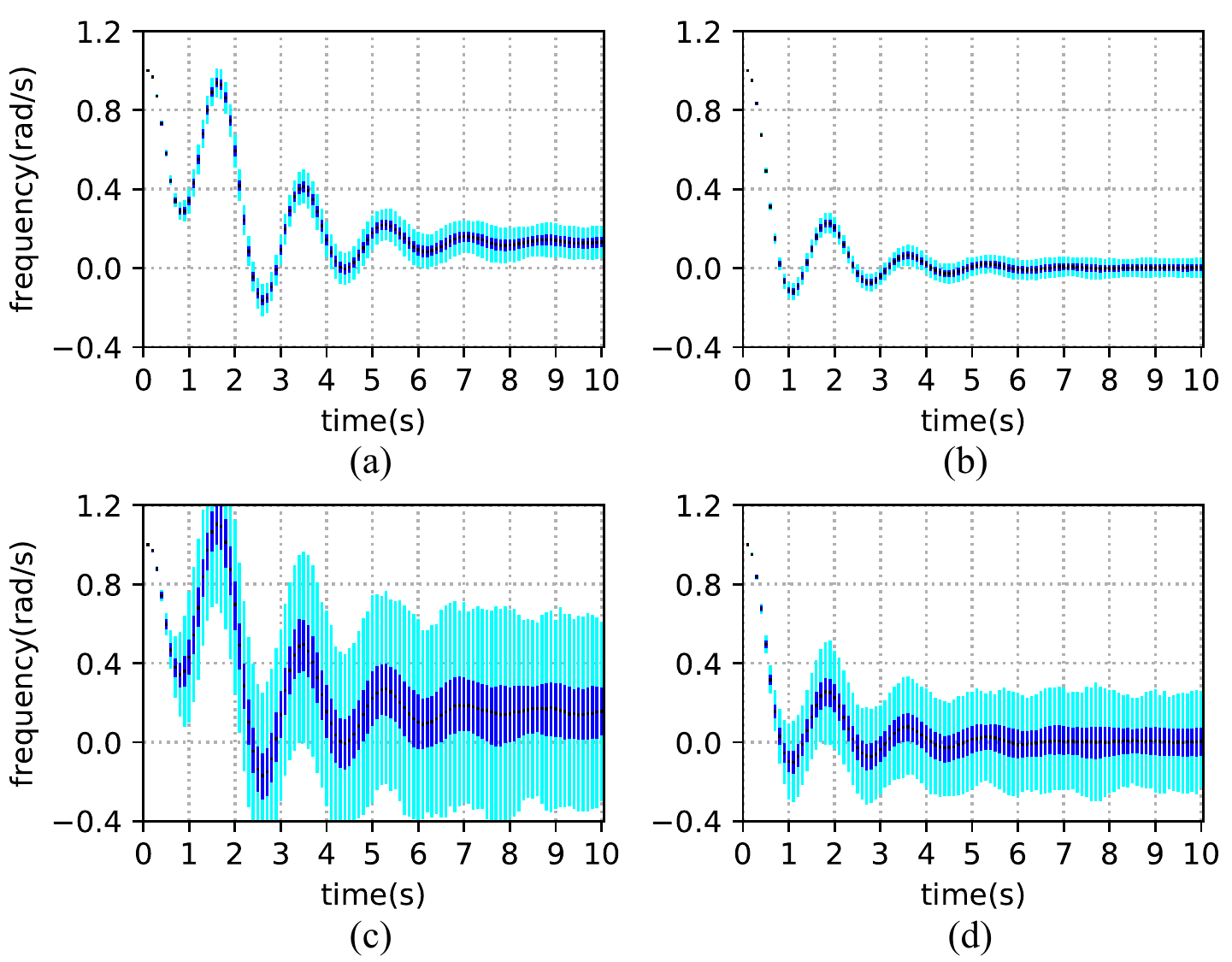}
\caption{
Box plots (1,000 test cases) of $\Delta \omega_{10}$, controlled by
(a) the standard LQG method under the worst-case distribution generated with $\theta = 0.5$,
(b) our minimax method under the worst-case distribution generated with $\theta = 0.5$,
(c) the standard LQG method under the worst-case distribution generated with $\theta = 1$,
and
(d) our minimax method under the worst-case distribution generated with $\theta = 1$.}
\label{fig:result}
\vspace{-0.1in}
\end{figure}

To model uncertainty in power injection or net demand,
a disturbance $w(t)$ is assumed to be added to the input $u(t)$. 
Then, $\Xi = B$.
For the quadratic cost function, we set $x^\top Q x= \frac{1}{2} \Delta \delta^\top (I_{10}-\frac{1}{10}\mathbf{1}_{10} \mathbf{1}_{10}^\top) \Delta \delta + \frac{1}{2} \Delta \omega^\top \Delta \omega$ and $R=I_{10}$,
where $I_{10}$ denotes the 10 by 10  identity matrix and $\mathbf{1}_{10}$ denotes the 10 dimensional vector of all ones. 
The system is discretized by a zero-order hold method with sample time $0.1$ seconds.
Suppose that the initial value of rotor speed $\Delta \omega_{10}$ is perturbed by $1$, $10$ samples of disturbances are generated according to the normal distribution $\mathcal{N}(0.02, 0.1^2 I)$,
 and the worst-case distribution in Corollary~\ref{cor:dist} is applied to the system in the finite-horizon setting with the number of stages $T=150$.\footnote{All the simulation codes and data   can be downloaded at https://github.com/hahakhkim/WassersteinLQ.}

\begin{table}
\caption{The settling time (in seconds) required for each generator to maintain the mean frequency less than $0.03$.}
\label{table:time}
\begin{center}
\begin{tabular}{p{4.5em}| *{10}{p{1.75em}}}
\hline
Gen \# & 1 & 2 & 3 & 4 & 5 & 6 & 7 & 8 & 9 & 10 \\
\hline
LQG & 4.8 & 4.9 & 4.7 & 4.9 & 5.1 & 4.2 & 4.4 & 4.1 & 3.9 & 7.1\\
Minimax & 1.9 & 3.2 & 3.2 & 2.9 & 3.1 & 2.7 & 2.7 & 3.7 & 3.4 & 4.5\\
\hline
\end{tabular}
\end{center}
\end{table}

\begin{figure*}[t!]
\centering
\includegraphics[scale=0.7]{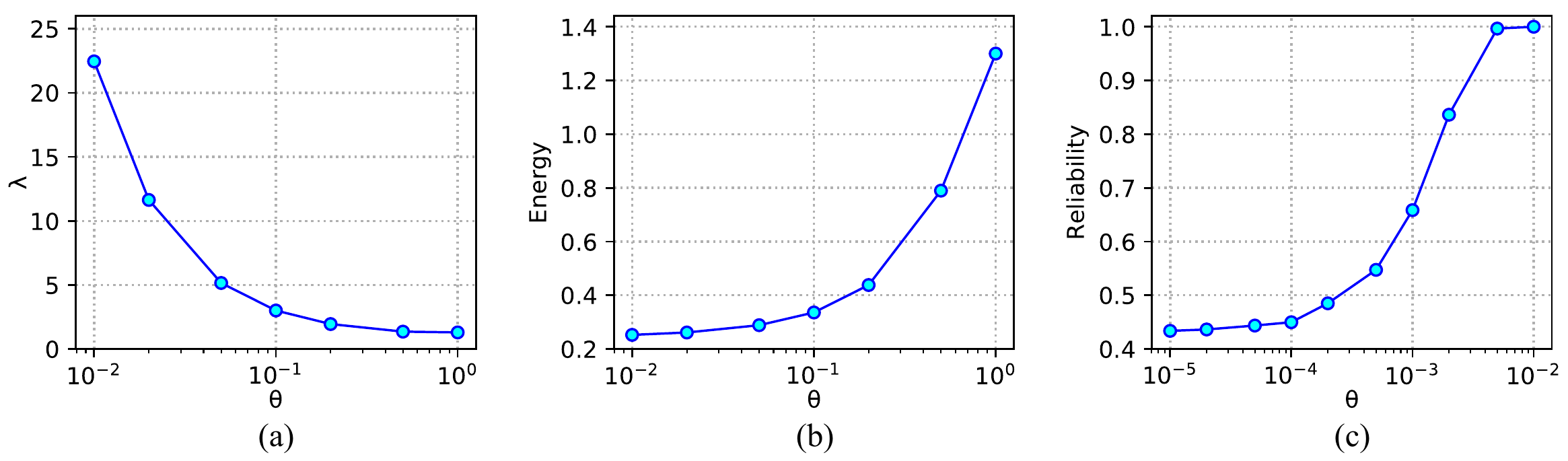}
\caption{(a) Optimal value of $\lambda$ depending on $\theta$,
(b) average control energy depending on $\theta$, and
(c) reliability depending on $\theta$.
}
\label{fig:result2}
\vspace{-0.1in}
\end{figure*}

Fig.~\ref{fig:result} shows the box plot of  frequency $\Delta \omega_{10}$, controlled by the standard LQG and the proposed minimax control methods.
The finite-horizon optimal policy \eqref{opt_policy}  is used, where the optimal penalty coefficient $\lambda^\star$ is obtained using Proposition~\ref{prop:lam} with the ambiguity set radius $\theta=0.5, 1.0$.
The  results demonstrate that our method significantly reduces the fluctuation of the frequency compared to the standard LQG method.
Additionally, the proposed control policy successfully drives the expected value of the system state to zero, while the standard LQG fails to do so.
The results also show that the size of the Wasserstein ambiguity set  or equivalently the value of $\theta$ plays an important role in the performance of our method.
As $\theta$ increases, the resulting policy with the penalty $\lambda^\star$ must guarantee the upper-bound of the cost function for a larger size of the ambiguity set.
Therefore, it is robust against a wider range of distributions, and the worst-case distribution is selected as a more extreme one.
As $\theta$ decreases, the worst-case distribution converges to the empirical distribution, and thus the robustness advantage of our policy over the standard LQG diminishes.
The settling time required for each generator to maintain the mean frequency less than $3\%$ of the initial deviation is shown in Table~\ref{table:time}, when the worst-case distribution with $\theta=0.5$ is applied to the system.
It takes $4.8$ seconds on average when using the standard LQG method, while the proposed minimax method requires $3.1$ seconds on average.

Fig.~\ref{fig:result2} (a) displays the optimal penalty parameter $\lambda^\star$ obtained using Proposition~\ref{prop:lam}, depending on the radius $\theta$.
The value of $\lambda^\star$ decreases as $\theta$ increases and eventually converges to the infimum of $\lambda^\star$ satisfying Assumption~\ref{ass:pen}.
This observation is  consistent with our intuition that the distributional robustness of the control policy can be tuned using the penalty parameter instead of $\theta$.

Fig.~\ref{fig:result2} (b) shows the average control energy required for our method depending on the value of $\theta$.
The control energy is measured for the first 5 seconds, i.e., $\sum_{t=0}^{49}\lVert u_t \rVert^2 / 50$, and  is averaged over 1,000 test cases. 
As shown in  Fig.~\ref{fig:result2}~(b), the required energy increases as $\theta$ increases.
If $\theta$ decreases, the required energy declines and eventually converges to the energy required for the standard LQG method.
This implies that a tradeoff between robustness and control energy exists in our method. 
Therefore, the value of $\theta$ should be properly selected based on the reliability of available data to balance robustness and control energy.

To test the out-of-sample performance of our control policy, 
 the \emph{reliability} 
$\mu^N \{  \hat{w} :  \frac{1}{T} \mathbb{E}^{\pi^\star_{\hat{w}}}_{w \sim \mu}$ $[\sum_{t=0}^{T-1}(x_t^\top Q x_t + u_t^\top R u_t) + x_T^\top Q_f x_T \mid x_0 = \bm{x} ] \leq \lambda \theta^2 + V (\bm{x}; \lambda)  \; \forall \bm{x} \in \mathbb{R}^n  \} $ is computed using 10,000 simulations with sample size $N = 10$.
 As shown in Fig.~\ref{fig:result2} (c), the reliability increases with $\theta$ as expected. 
 More specifically,  the reliability sharply increases in $[10^{-4}, 3\times 10^{-3}]$ and then saturates as $\theta$ increases further.
Given that the control energy also increases with $\theta$, 
it may be reasonable to choose $\theta \approx 3\times 10^{-3}$ in this problem to attain a sufficiently robust policy, which is not overly conservative.

\section{Conclusions}

We have presented a minimax LQ control method  with a Wasserstein penalty to address the issue of ambiguity inherent in practical stochastic systems.
Our method has several salient features including $(i)$ a closed-form expression of the optimal policy pair, $(ii)$  the convergence of a Riccati equation to the unique symmetric PSD solution to the corresponding ARE, $(iii)$ closed-loop stability, $(iv)$ distributional robustness, and $(v)$ an out-of-sample performance guarantee. 
The relation to the $H_\infty$-method indicates that our method may open an exciting avenue for future research that connects stochastic and robust control from the perspective of distributional robustness.
Moreover, it remains as future work to address partial observability and extensions to continuous-time settings.

\appendix
\section{Proofs}
\subsection{Proof of Lemma~\ref{lem:sol}}\label{app:lem:sol}

\begin{proof}
The function $w \mapsto V_{t+1}(A \bm{x} + B \bm{u} + \Xi w)-\lambda \lVert \hat{w}^{(i)}_t - w \rVert^2$ is strictly concave quadratic  under the assumption $\lambda I - \Xi^\top P_{t+1} \Xi \succ 0$.
Differentiating it with respect to $w$, we obtain the following optimality condition:
\begin{equation}\label{eq1}
w_t^{\star, (i)} = \frac{1}{2\lambda} \Xi^\top V_{t+1}' (A \bm{x} + B\bm{u} + \Xi w_t^{\star, (i)}  )+ \hat{w}_t^{(i)},
\end{equation}
which directly yields \eqref{w_opt}.
To solve the outer minimization problem in \eqref{eq:backward},
we first differentiate the outer objective function with respect to $\bm{u}$ to obtain that
\begin{equation} \nonumber
\begin{split}
&2R \bm{u} + \frac{1}{N} \sum_{i=1}^N \bigg [
\bigg (B + \Xi \frac{\partial w_t^{\star, (i)}}{\partial \bm{u}}  \bigg )^\top  V_{t+1}' (A\bm{x} + B\bm{u} + \Xi w_t^{\star, (i)} )
+ 2\lambda \frac{\partial w_t^{\star, (i)}}{\partial \bm{u}}^\top (\hat{w}_t^{(i)} - w_t^{\star, (i)})
\bigg ]\\
&= 2R \bm{u} + 2 B^\top g_t (\bm{u}),
\end{split}
\end{equation}
where 
\begin{equation} \label{g_eq}
\begin{split}
g_t (\bm{u}) &:= \frac{1}{2N} \sum_{i=1}^N V_{t+1}' (A\bm{x} + B\bm{u} + \Xi w_t^{\star, (i)}(\bm{u}) )\\
&= P_{t+1} \bigg ( 
A \bm{x} + B\bm{u} + \frac{1}{N} \sum_{i=1}^N \Xi w_t^{\star, (i)}(\bm{u}) \bigg ) + r_{t+1}.
\end{split}
\end{equation}
The Hessian of the outer objective function with respect to $\bm{u}$ is then given by
\[
2 \big [ 
R + B^\top P_{t+1} B + B^\top P_{t+1} \Xi (\lambda I - \Xi^\top P_{t+1} \Xi)^{-1} \Xi^\top P_{t+1} B
\big ],
\]
which is positive definite under the assumption on the penalty parameter. 
Thus, the outer objective function is strictly convex, and it has a unique minimizer, $\bm{u}^\star$. 
Equating the derivative to zero yields
\begin{equation}\label{u_eq}
\bm{u}^\star = -R^{-1} B^\top g_t^\star,
\end{equation}
where $g_t^\star := g_t (\bm{u}^\star)$.
By the definition of $g_t$ and~\eqref{eq1}, 
\begin{equation*}
\begin{split}
g_t^\star &= P_{t+1} \bigg ( 
A \bm{x} - BR^{-1} B^\top g_t^\star  + \frac{1}{\lambda} \Xi \Xi^\top g_t^\star+ \Xi \sum_{i=1}^N \frac{\hat{w}_t^{(i)}}{N} \bigg ) + r_{t+1},
\end{split}
\end{equation*}
which yields the following expression of $g_t^\star$:
\begin{equation}\label{g_eq2}
\begin{split}
g_t^\star = \bigg  (
I + P_{t+1} B R^{-1} B^\top  -  \frac{1}{\lambda} P_{t+1} \Xi \Xi^\top
\bigg )^{-1}\bigg( P_{t+1} A \bm{x} + P_{t+1}\Xi \sum_{i=1}^N \frac{\hat{w}_t^{(i)}}{N} + r_{t+1} \bigg) .
\end{split}
\end{equation}
Note that $I + P_{t+1} B R^{-1} B^\top  -  \frac{1}{\lambda} P_{t+1} \Xi \Xi^\top$ must be invertible by the uniqueness of $\bm{u}^\star$.
\end{proof}

\subsection{Proof of Theorem~\ref{thm:fin}}\label{app:thm:fin}

\begin{proof}
We use mathematical induction to show that $V_t(\bm{x}) = \bm{x}^\top P_t \bm{x} + 2r_t^\top \bm{x} + z_t$. 
For $t = T$, the statement is true by the definition of $P_T$, $r_T$, and $z_T$. 
Suppose that the induction hypothesis is true for $t+1$, i.e., $V_{t+1}(\bm{x}) = \bm{x}^\top P_{t+1} \bm{x} +  2r_{t+1}^\top \bm{x} + z_{t+1}$.
Recall that $g_t^\star :=g_t (\bm{u}^\star)$, where $g_t$ is given as~\eqref{g_eq}. 
Differentiating \eqref{eq:backward} with respect to $\bm{x}$ and using
 \eqref{eq1} and \eqref{u_eq}, we obtain
\begin{equation}\nonumber
\begin{split}
	& V_t' (\bm{x}) = 2 Q\bm{x} + 2 \frac{\partial \bm{u}^\star}{\partial \bm{x}}^\top R  \bm{u}^\star\\
	& +  \frac{1}{N} \sum_{i=1}^N \bigg [\bigg ( A + B \frac{\partial \bm{u}^\star}{\partial \bm{x}} + \Xi \frac{\partial \bm{w}^{\star, (i)}}{\partial \bm{x}} \bigg )^\top V_{t+1}' ( A \bm{x} + B \bm{u}^\star + \Xi \bm{w}^{\star, (i)} )+2 \lambda \frac{\partial \bm{w}^{\star, (i)}}{\partial \bm{x}}^\top (\hat{w}_t^{(i)} - \bm{w}^{\star, (i)} ) \bigg ]\\
	&= 2 Q\bm{x} + A^\top \frac{1}{N} \sum_{i=1}^N V_{t+1}' ( A \bm{x} + B \bm{u}^\star + \Xi \bm{w}^{\star, (i)} )  \\
	& + \frac{\partial \bm{u}^\star}{\partial \bm{x}}^\top  \bigg [2R \bm{u}^\star  + B^\top \frac{1}{N} \sum_{i=1}^N V_{t+1}' ( A \bm{x} + B \bm{u}^\star + \Xi \bm{w}^{\star, (i)} ) \bigg ]\\
	&=2Q\bm{x} + 2A^\top g_t^\star,
 \end{split}
 \end{equation}
 where $\bm{u}^\star$ is given as~\eqref{u_opt} and $\bm{w}^{\star, (i)}$ is given as~\eqref{w_opt} with $\bm{u} := \bm{u}^\star$.
  Replacing $g_t^\star$ with \eqref{g_eq2} yields
  \begin{equation}\label{eq_Q}
   Q\bm{x} + A^\top g_t^\star = P_t \bm{x} +  r_t
  \end{equation}
  by the recursion for $P_t$ in the Riccati equation~\eqref{ric}.
  Thus, 
  \[
  \frac{1}{2} V_t'(\bm{x}) = P_t \bm{x} + r_t,
  \]
  which implies that
 \[
  V_t (\bm{x}) = \bm{x}^\top P_t \bm{x} + 2 r_t^\top \bm{x} + z'_t,
 \]
for some constant  $z'_t \in \mathbb{R}$.

Plugging $\bm{u}^\star = K_t \bm{x} + L_t$ into \eqref{eq:backward} yields
\begin{equation}\label{eq:zprime}
\begin{split}
 z'_t &=  L_t^\top R L_t + z_{t+1}\\
& + \frac{1}{N} \sum_{i=1}^N \big[ ( BL_t + \Xi \bm{w}^{\star, (i)} )^\top  P_{t+1}( BL_t + \Xi \bm{w}^{\star, (i)} ) + 2 r_{t+1}^\top (BL_t + \Xi \bm{w}^{\star, (i)})-\lambda \lVert \hat{w}^{(i)}_t - \bm{w}^{\star, (i)} \rVert^2 \big ].
 \end{split}
 \end{equation}
For simplicity, let $\alpha_i:= \bm{w}^{\star, (i)} - \hat{w}_t^{(i)}$ and $\beta_i:= BL_t + \Xi  \hat{w}_t^{(i)}$.
Then, each term in the summation can be written as
\begin{equation}\nonumber
\begin{split}
&( \Xi \alpha_i + \beta_i)^\top  P_{t+1} (  \Xi \alpha_i + \beta_i ) + 2 r_{t+1}^\top (  \Xi \alpha_i + \beta_i) -\lambda \alpha_i^\top \alpha_i \\
&= \beta_i^\top P_{t+1} \beta_i + 2 \alpha_i^\top \Xi^\top (P_{t+1} \beta_i + r_{t+1})  + 2r_{t+1}^\top \beta_i -  \alpha_i^\top (\lambda I - \Xi^\top P_{t+1} \Xi) \alpha_i.
 \end{split}
 \end{equation}
It follows from \eqref{w_opt} that the constant part of $\alpha_i$ (with respect to $\bm{x}$) is given by $(\lambda I - \Xi^\top P_{t+1} \Xi )^{-1} \Xi^\top$ $(P_{t+1} \beta_i + r_{t+1})$.
Plugging it into the equality above, we have
\begin{equation}\nonumber
\begin{split}
& \beta_i^\top P_{t+1} \beta_i +  (P_{t+1} \beta_i + r_{t+1})^\top \Xi (\lambda I - \Xi^\top P_{t+1} \Xi )^{-1}  \Xi^\top  (P_{t+1} \beta_i + r_{t+1}) +  2r_{t+1}^\top  \beta_i\\
&= \beta_i^\top P_{t+1} (I -  \frac{1}{\lambda}  \Xi \Xi^\top P_{t+1})^{-1} \beta_i + 2r_{t+1}^\top \bigg [ I - \frac{1}{\lambda} \Xi \Xi^\top P_{t+1} \bigg ]^{-1} \beta_i + r_{t+1}^\top \Xi (\lambda I - \Xi^\top P_{t+1} \Xi)^{-1} \Xi^\top r_{t+1}.
 \end{split}
 \end{equation}
Substituting $\beta_i$ with $BL_t + \Xi  \hat{w}_t^{(i)}$, and $L_t$ with $-R^{-1} B^\top  (I + P_{t+1} \Phi)^{-1} ( P_{t+1}\Xi \bar{w}_t + r_{t+1})$, \eqref{eq:zprime} can be expressed as
\begin{equation}\nonumber
\begin{split}
z'_t &= z_{t+1}  + \mathrm{tr}   [ (I -  \Xi^\top P_{t+1} \Xi/\lambda )^{-1} \Xi^\top P_{t+1} \Xi \Sigma_t   ] \\
&+\bar{w}_t^\top \Xi^\top  [(I+P_{t+1}\Phi)^{-1}-(I- P_{t+1} \Xi \Xi^\top/\lambda)^{-1}] P_{t+1} \Xi \bar{w}_t  + (2  \bar{w}_t^\top \Xi^\top - r_{t+1}^\top \Phi)(I+P_{t+1} \Phi)^{-1} r_{t+1},
\end{split}
\end{equation}
where we have omitted the detailed algebra.
Finally, by the recursion for $z_t$ in \eqref{ric}, we deduce that $z'_t = z_t$.
This  completes our inductive argument. 
 Lastly, It follows from Lemma~\ref{lem:sol} that an optimal policy must be unique and it is obtained as \eqref{opt_policy}. 
 \end{proof}
 
\subsection{Proof of Proposition~\ref{prop:are}} 

\begin{proof}
In the standard LQG, it is well known that if $(A, B)$ is stabilizable, the Riccati equation~\eqref{ric_lqg} has a bounded limiting solution, which coincides with a symmetric PSD solution to an associated ARE~\cite[Theorem 2.4-1]{Lewis2012}.
Note that the ARE~\eqref{are} can be rewritten as
\begin{equation}\nonumber
\begin{split}
P &= Q + A^\top ( I + P\Phi)^{-1} PA\\
&= Q + A^\top ( I + P \underbrace{\sqrt{\Phi}}_{B'} \underbrace{I^{-1}}_{R'^{-1}}  \sqrt{\Phi}^\top )^{-1} PQ\\
&= Q + A^\top P A - A^\top P B' (R' + B'^\top P B')^{-1} B'^\top PA,
\end{split}
\end{equation}
which is in the form of the standard ARE.
Thus, our ARE~\eqref{are} is obtained by replacing $(A, B)$ with $(A, \sqrt{\Phi})$ in the ARE for the standard LQG, and the result follows.
\end{proof}

\subsection{Proof of Lemma~\ref{lem:sol2}}\label{app:lem:sol2}

 \begin{proof}
 Let $P$ be a solution to the equation $P-Q = A^\top P (I + \Phi P)^{-1} A$. Let $E:=(I +  \Phi P)^{-1} A$ be decomposed as $E=U_1 D U_1^{-1}$, where $D$ is a Jordan normal form.
Then, we have $P-Q = A^\top P U_1 D U_1^{-1}$.
Let $U_2:=PU_1$. Then, we obtain 
\[
U_2 -Q U_1 = A^\top U_2 D.
\]
Since $A = (I+ \Phi P)E=(I+ \Phi U_2 U_1^{-1}) U_1DU_1^{-1}$, we have 
\[
AU_1 = U_1D +  \Phi U_2 D.
\]
Therefore, we obtain that  $F \begin{bmatrix} U_1 \\ U_2 \end{bmatrix}  = G\begin{bmatrix} U_1 \\ U_2 \end{bmatrix} D$. 
This implies that a solution to the ARE~\eqref{are} is expressed as $P=U_2U_1^{-1}$ and $\begin{bmatrix} U_1 \\ U_2 \end{bmatrix}$ solves generalized eigenvalue problem.
\end{proof}

\subsection{Proof of Lemma~\ref{lem:stable}}\label{app:lem:stable}

\begin{proof}
Suppose first that $F \begin{bmatrix} U_1 \\ U_2 \end{bmatrix} = G \begin{bmatrix} U_1 \\ U_2 \end{bmatrix} \Lambda$, where $\Lambda$ is a Jordan normal form.
Then, $A = (I + \Phi U_2 U_1^{-1}) U_1 \Lambda U_1^{-1}$.
It follows from the ARE~\eqref{are} that
\begin{equation*}
\begin{split}
    U_2 U_1^{-1} &= Q + A^\top U_2 U_1^{-1} (I+\Phi U_2 U_1^{-1})^{-1} A\\
    &= Q + (U_1^{-H} \Lambda^H U_1^H + U_1^{-H} \Lambda^H U_2^H \Phi^H ) U_2 \Lambda U_1^{-1}.
\end{split}
\end{equation*}
This is a discrete-time Lyapunov equation of the form
\begin{equation}\label{lyap}
P = \bar{A}^H P \bar{A} + \bar{Q},
\end{equation}
 where $P= U_2 U_1^{-1}$, $\bar{A}:= U_1 \Lambda U_1^{-1}$, and $\bar{Q} := Q+(U_2 \Lambda U_1^{-1})^H \Phi U_2 \Lambda U_1^{-1}$.
Note that $\bar{Q} \succeq 0$ since $\Phi \succeq 0$ under Assumption~\ref{ass:W}.
By the theory of Lyapunov equations, 
we conclude that $P \succeq 0$ since
$\bar{Q} \succeq 0$ and $\bar{A}$ is stable.

We  now  assume that $P \succeq 0$. Suppose that  $\bar{Q}  \succeq 0$, and $\bar{A}$ has an unstable eigenvalue, i.e., $\bar{A} v = \gamma v$, where $\vert \gamma \vert \geq 1$.
Pre-multiplying $v^H$ and post-multiplying $v$ on  both sides of the Lyapunov equation~\eqref{lyap}, we obtain $(\gamma^* \gamma -1) v^H P v + v^H \bar{Q} v = 0$.
Then, ${\bar{Q}}^{1/2} v=0$, which leads to ${Q}^{1/2} v =  {\Phi}^{1/2} U_2 \Lambda U_1^{-1} v=0$ and $Av = (I+\Phi P)\bar{A}v = \gamma v$.
This contradicts Assumption \ref{ass:ob}.
Therefore, if $P  \succeq 0$ and $\bar{Q}  \succeq 0$, then $\bar{A}$ must be stable.
Since $\bar{A}$ and $\Lambda$ have the same spectrum, the result follows.
\end{proof}

\subsection{Proof of Proposition~\ref{prop:r}}

\begin{proof}
The recursion of $r_t$ in \eqref{ric} becomes
$r_t =  A^\top  (I + P_{ss} \Phi )^{-1} (P_{ss} \Xi  \bar{w} + r_{t+1} )$
as $P_t$ converges to $P_{ss}$.
When $A^\top (I+P_{ss} \Phi)^{-1}$ is stable, $r_t$  must converge  to \eqref{eq:oss}.
Thus, it suffices to show that $A^\top (I+P_{ss} \Phi)^{-1}$ is stable.
If follows from \eqref{gen2} that $A = U_1 \Lambda U_1^{-1} + \Phi U_2 \Lambda U_1^{-1}$ and
\begin{equation*}
\begin{split}
(I + \Phi P_{ss})^{-1}A 
&= (I + \Phi U_2 U_1^{-1})^{-1} (U_1 \Lambda U_1^{-1} + \Phi U_2 \Lambda U_1^{-1})\\
& = U_1 \Lambda U_1^{-1},
\end{split}
\end{equation*}
which implies that $(I + \Phi P_{ss})^{-1}A$ and $\Lambda$ have the same spectrum.
Since $\Lambda$ has $n$ stable eigenvalues, $(I + \Phi P_{ss})^{-1}A$ is stable, and so is its transpose  $A^\top (I+P_{ss} \Phi)^{-1}$.
\end{proof}

\subsection{Proof of Proposition~\ref{prop:avgcost}}

\begin{proof}
The steady-state average cost is computed as
\begin{equation*}
\begin{split}
& \limsup_{T \rightarrow \infty} \frac{1}{T} ( \bm{x}^\top P_{0} \bm{x} + 2 r_{0}^\top \bm{x} +  z_{0} )= \limsup_{T \rightarrow \infty} \frac{1}{T} (\bm{x}^\top P_{ss} \bm{x} + 2 r_{ss}^\top \bm{x} +  z_{0}) =\limsup_{T \rightarrow \infty} \frac{z_0}{T}.
\end{split}
\end{equation*}
By the recursion for $z_t$ in \eqref{ric}, this cost can be expressed as
\begin{equation*}
\begin{split}
	\limsup_{T \rightarrow \infty} \frac{1}{T} \sum_{t=0}^{T-1} &\bigg[ \mathrm{tr}   [ (I -  \Xi^\top P_{t+1} \Xi/\lambda )^{-1} \Xi^\top P_{t+1} \Xi \Sigma  ]\\
&+\bar{w}^\top \Xi^\top  [(I+P_{t+1}\Phi)^{-1}-(I- P_{t+1} \Xi \Xi^\top/\lambda)^{-1}] P_{t+1} \Xi \bar{w} \\
& + (2  \bar{w}^\top \Xi^\top - r_{t+1}^\top \Phi)(I+P_{t+1} \Phi)^{-1} r_{t+1}\bigg],
\end{split}
\end{equation*}
which is equal to the value given in the statement.
\end{proof}

\subsection{Proof of Proposition~\ref{prop:acoe}}

\begin{proof}
We use Lemma~\ref{lem:sol} with $V_{t+1} \equiv h$, setting $P_{t+1} = P_{ss}$ and $r_{t+1} = r_{ss}$ and $z_{t+1} = 0$. 
Then, $V_t (\bm{x}) = \bm{x}^\top P_t \bm{x} + 2r_{t}^\top \bm{x} + z_t$ in \eqref{eq:backward} satisfies  \eqref{ric} with $P_{t+1} = P_{ss}$ and $r_{t+1} = r_{ss}$ and $z_{t+1} = 0$:
\begin{equation*}
\begin{split}
P_t &= Q + A^\top (I + P_{ss}\Phi )^{-1} P_{ss} A\\
r_t &=  A^\top (I + P_{ss} \Phi)^{-1} ( P_{ss}\Xi \bar{w} + r_{ss} )\\
z_t &= \mathrm{tr}   [ (I -  \Xi^\top P_{ss} \Xi/\lambda )^{-1} \Xi^\top P_{ss} \Xi \Sigma   ]\\
&+\bar{w}^\top \Xi^\top  [(I+P_{ss}\Phi)^{-1}-(I- P_{ss} \Xi \Xi^\top/\lambda)^{-1}] P_{ss} \Xi \bar{w}\\
& + (2  \bar{w}^\top \Xi^\top - r_{ss}^\top \Phi)(I+P_{ss} \Phi)^{-1} r_{ss}.
\end{split}
\end{equation*}
It follows from the ARE~\eqref{are} that $P_t = P_{ss}$.
By the definition of $r_{ss}$ in \eqref{eq:oss}, 
\begin{equation} \nonumber
\begin{split}
 P_{ss}\Xi \bar{w} + r_{ss} &= P_{ss}\Xi \bar{w}  + [I - A^\top (I+P_{ss} \Phi)^{-1}]^{-1} A^\top (I+P_{ss} \Phi)^{-1} P_{ss} \Xi \bar{w}\\
&= [I - A^\top (I+P_{ss} \Phi)^{-1}]^{-1}  P_{ss} \Xi \bar{w}.
\end{split}
\end{equation}
Thus, we deduce that $r_t = r_{ss}$. 
Moreover, Proposition~\ref{prop:avgcost} implies that $z_t = \rho$. 
Putting these results together, we conclude that $V_t (\bm{x}) = \bm{x}^\top P_{ss} \bm{x} + r_{ss}^\top \bm{x} + \rho = h(\bm{x}) + \rho$. 
Therefore, the equality~\eqref{eq:acoe} holds.
The optimality of $(\pi_{ss}^\star (\bm{x}), \gamma_{ss}^\star (\bm{x}))$ also follows from Lemma~\ref{lem:sol}.
\end{proof}

\subsection{Proof of Theorem~\ref{thm:inf}}

\begin{proof}
{\it (a)} Consider any single-stage policy pair $(\pi_t, \gamma_t)$.
We define the mapping $\mathcal{T}^{\pi_t, \gamma_t}$ as
\begin{equation*}
\begin{split}
	(\mathcal{T}^{\pi_t, \gamma_t} h) (\bm{x}) &:=  \bm{x}^\top Q \bm{x} + u_t^\top R u_t - \lambda W_2({\mu}_t, \nu)^2\\
 &+\int_{\mathbb{R}^k} h (A \bm{x} + B u_t + \Xi w) \mathrm{d} \mu_t (w), \;\; u_t = \pi_t (\bm{x}), \mu_t = \gamma_t (\bm{x}).
\end{split}
\end{equation*}
It follows from Proposition~\ref{prop:acoe} that
\[
(\mathcal{T}^{\pi_t, \gamma^\star_{ss}} h) (\bm{x}) \geq (\mathcal{T}^{\pi^\star_{ss}, \gamma^\star_{ss}} h) (\bm{x}) = \rho + h(\bm{x}).
\]
Fix a policy $\pi := (\pi_0, \pi_1, \ldots) \in \Pi$ and an arbitrary positive integer $T$.
We first note that
\[
(\mathcal{T}^{\pi_{T-1}, \gamma^\star_{ss}} h)(\bm{x}) \geq \rho + h(\bm{x}). 
\]
By the monotonicity of the mapping $\mathcal{\mathcal{T}}^{\pi_{T-2}, \gamma^\star_{ss}}$, we have
\begin{equation*}
\begin{split}
(\mathcal{T}^{\pi_{T-2}, \gamma^\star_{ss}} \mathcal{T}^{\pi_{T-1}, \gamma^\star_{ss}} h)(\bm{x})
&\geq \rho + (\mathcal{T}^{\pi_{T-2}, \gamma^\star_{ss}}h)(\bm{x})\\
&\geq 2\rho +h(\bm{x}).
\end{split}
\end{equation*}
Recursively applying this inequality yields
\[
	(\mathcal{T}^{\pi_{0}, \gamma^\star_{ss}} \cdots \mathcal{T}^{\pi_{T-1}, \gamma^\star_{ss}} h)(\bm{x}) \geq T\rho +h(\bm{x}).
\]
Dividing both sides by $T$ and letting $T$ tend to $\infty$, we have 
\begin{equation}\label{eq:saddle1}
	 \tilde{J}^{\lambda}_{\bm{x}, \infty}(\pi, \gamma_{ss}^\star;h) \geq \rho,
\end{equation}
which holds for any $\pi \in \Pi$. 
By the same argument with $\pi_t := \pi_{ss}^\star$, 
we can also show that $\tilde{J}^{\lambda}_{\bm{x}, \infty}(\pi_{ss}^\star, \gamma; h) \leq \rho$ for all $\gamma \in \Gamma$.

{\it (b)} Inequalities \eqref{eq:saddle} imply that  $(\pi_{ss}^\star, \gamma_{ss}^\star)$ is an optimal policy pair of the minimax problem~\eqref{eq:avgopt}. Furthermore, we deduce that $\tilde{J}^{\lambda}_{\bm{x}, \infty}(\pi_{ss}^\star, \gamma_{ss}^\star;h) = \rho$. 

{\it (c)} Under the conditions \eqref{eq:avg} and \eqref{eq:avg2}, the $h$ term in \eqref{eq:saddle} can be ignored, and thus the following inequalities hold:
 \begin{equation*}
{J}^{\lambda}_{\bm{x}, \infty}(\pi_{ss}^\star, \gamma)  \leq \rho
 \leq {J}^{\lambda}_{\bm{x}, \infty}(\pi, \gamma_{ss}^\star) \quad \forall (\pi, \gamma) \in \bar{\Pi} \times \bar{\Gamma}.
\end{equation*}
Using the argument in Part {\it (b)},
we conclude that $(\pi_{ss}^\star, \gamma_{ss}^\star)$ is an optimal policy pair and $\rho$ is the optimal value of the problem~\eqref{eq:acopt}.
\end{proof}

\subsection{Proof of Theorem~\ref{thm:stable1}}\label{app:thm:stable1}

\begin{proof}
Let $\bar{x}_0 := x_0^\star$ and $\bar{x}_t := \mathbb{E} [x_t^\star]$ for $t = 1, 2, \ldots$, where $x_t^\star$ denotes the closed-loop system state under the optimal policy in Corollary~\ref{cor:inf}.
Then, the mean-state system is given by
\begin{equation*}
\begin {split}
\bar{x}_{t+1} 
&=  (A+ B K_{ss}) \bar{x}_t + B L_{ss} + \frac{\Xi}{N} \sum_{i=1}^N w^{\star, (i)} (\bar{x}_t ).
\end{split}
\end{equation*}
It follows from \eqref{eq1} and \eqref{g_eq2} that 
\begin{equation*}
\begin {split}
\frac{\Xi}{N} \sum_{i=1}^N w^{\star, (i)} (\bar{x}_t ) &= \frac{1}{\lambda} \Xi \Xi^\top g_t^\star + \Xi \bar{w}\\
&=\frac{1}{\lambda} \Xi \Xi^\top (I + P_{ss} \Phi)^{-1} (P_{ss} A \bm{x} + P_{ss} \Xi \bar{w} + r_{ss}) + \Xi \bar{w}.
\end{split}
\end{equation*} 
By this equality and the definition of $L_{ss}$ and $r_{ss}$, the mean-state dynamics can be rewritten as
\begin{equation*}
\begin {split}
\bar{x}_{t+1} =& (I + \Phi P_{ss})^{-1}A  \bar{x}_{t} + (I-\Phi (I+P_{ss} \Phi - A^\top)^{-1} P_{ss} ) \Xi \bar{w}.
\end{split}
\end{equation*}
It follows from the proof of Proposition~\ref{prop:r} that the gain $(I + \Phi P_{ss})^{-1}A $ is stable,
and therefore the expected state converges to 
$[I - (I + \Phi P_{ss})^{-1}A ]^{-1}
[ I-\Phi (I+P_{ss} \Phi - A^\top)^{-1} P_{ss} ]\Xi \bar{w}$.
\end{proof}

\subsection{Proof of Theorem~\ref{thm:stable2}}\label{app:thm:stable2}

\begin{proof}
Since $A+BK_{ss}$ is independent of $\nu$, without loss of generality, we let $\nu \equiv \delta_0$.
Consider the policy $\gamma'$ that selects $\mu_t = \delta_0 = \nu$ for all $t \geq 0$.
Under the policy pair $(\pi_{ss}^\star, \gamma')$, the closed-loop system is given by $x_{t+1} = (A + BK_{ss})x_t$ for all $t \geq 0$.
Since $r_{ss}=0$ in this case, we obtain that $h(\bm{x}) = \bm{x}^\top P_{ss} \bm{x}$ and $\rho=0$.
Since $h(\bm{x}) \geq 0$ for all $\bm{x}$, $\bar{\Gamma}$ is equivalent to $\Gamma$, and thus $\gamma' \in \bar{\Gamma}$.

It follows  from Theorem~\ref{thm:inf} that
\[
	\tilde{J}^{\lambda}_{\bm{x}, \infty} (\pi^\star_{ss}, \gamma'; h) \leq \rho = 0.
\]
Recall that $h(\bm{x}) = \bm{x}^\top P_{ss} \bm{x}$ and $W_2(\mu_t, \nu) = 0$ in the above setting. Thus,
\begin{equation*}
\limsup_{T \rightarrow \infty} \frac{1}{T} \mathbb{E}^{\pi^\star_{ss}, \gamma'} \bigg[ \sum_{t=0}^{T-1}(x_t^\top Q x_t + u_t^\top R u_t ) + x_T^\top P_{ss} x_T \bigg\vert x_0 = \bm{x} \bigg]
\end{equation*}
is less than or equal to $\rho = 0$.
Since $Q, P_{ss} \succeq 0$ and $R \succ 0$,   $\tilde{J}^{\lambda}_{\bm{x}, \infty} (\pi^\star_{ss}, \gamma'; h)=  0$ and $\limsup$ is replaced by $\lim$. 
This implies that
\begin{equation*}
    \lim_{t \rightarrow \infty} (x_t^\top Q x_t + u_t^\top R u_t) = 0,
\end{equation*}
which implies that ${Q}^{1/2} x_t \rightarrow 0$ and $u_t \rightarrow 0$ as $t \to \infty$.

Note that the linear system can be expressed as 
\[
x_{t+k} =  A^k x_t + \sum_{l=0}^{k-1} A^{k-1-l} B u_{t+l}, \quad k \geq 1.
\]
It follows from the triangle inequality that
\begin{equation*}
\rVert Q^{1/2} A^k x_t \rVert^2  \leq \lVert Q^{1/2} x_{t+k} \rVert^2 + \sum_{l=0}^{k-1}\Vert  Q^{1/2} A^{k-1-l} B u_{t+l} \rVert^2.
\end{equation*}
Recall that $Q^{1/2}x_t$ and $u_t$ converge to $0$ as $t \to \infty$.
Thus, for any $\epsilon>0$, there exists $T(\epsilon)$ such that
\[
\sum_{k=0}^{n-1} \rVert Q^{1/2} A^k x_{t} \rVert^2 \leq \epsilon^2, \quad t>T(\epsilon).
\]
The left-hand side is the squared Euclidean norm of the product of the observability matrix and $x_t$. 
By the observability of $(A, \sqrt{Q})$, the observability matrix has a full rank.
Thus, $\lVert x_t \rVert \leq \epsilon/\sigma_{\min}$ and $x_t$ converges to $0$, where $\sigma_{\min}$ is the smallest singular value of the observability matrix.
Hence, the closed-loop system $x_{t+1} = (A + BK_{ss})x_t$ is asymptotically stable and the mean-state system with $\pi_{ss}^\star$ is BIBO stable. 
\end{proof}

\subsection{Proof of Lemma~\ref{lem:bd}}

\begin{proof}
Fix $\pi \in \Pi$.
Let $p^\star := \sup_{\gamma \in \Gamma_{\mathcal{D}}} J_{\bm{x}} (\pi, \gamma)$ and $d^\star := \inf_{\lambda \geq 0} \sup_{\gamma \in \Gamma} (\lambda \theta^2 + J_{\bm{x}}^\lambda (\pi, \gamma))$. 
For any $\varepsilon > 0$, there exists $\gamma^\varepsilon \in \Gamma_{\mathcal{D}}$ such that
\begin{equation}\label{ie1}
p^\star - \varepsilon < J_{\bm{x}} (\pi, \gamma^\varepsilon).
\end{equation}
Since $\gamma^\varepsilon \in \Gamma_{\mathcal{D}}$, we have
\begin{equation}\nonumber
\begin{split}
&\lambda \theta^2 + J_{\bm{x}}^\lambda (\pi, \gamma^\varepsilon) \\
&\geq 
\frac{1}{T}\mathbb{E}^{\pi, \gamma^\varepsilon}\bigg [  \sum_{t=0}^{T-1} \lambda W_2 ( \mu_t, \nu_t)^2 + x_T^\top Q_f x_T  + \sum_{t=0}^{T-1} (x_t^\top Q x_t + u_t^\top R u_t - \lambda W_2 (\mu_t, \nu_t)^2)~\bigg\vert~x_0 = \bm{x} 
\bigg ]\\
&= J_{\bm{x}} (\pi, \gamma^\varepsilon).
\end{split}
\end{equation}
Minimizing both sides with respect to $\lambda \geq 0$ yields
\begin{equation}\label{ie2}
d^\star \geq J_{\bm{x}} (\pi, \gamma^\varepsilon).
\end{equation}
Combining inequalities~\eqref{ie1} and~\eqref{ie2}, we obtain
\[
p^\star -\varepsilon < d^\star.
\]
Since this inequality holds for any $\varepsilon > 0$, we conclude that
$p^\star \leq d^\star$.
\end{proof}

\subsection{Proof of Theorem~\ref{thm:bd}}

\begin{proof}
It follows from Lemma~\ref{lem:bd} that 
\begin{equation}\nonumber
\begin{split}
\sup_{\gamma \in \Gamma_{\mathcal{D}}} J_{\bm{x}} ( \pi^{\star,  \lambda}, \gamma ) &\leq \inf_{\lambda' \geq 0} \sup_{\gamma \in \Gamma} \big (\lambda' \theta^2 + J_{\bm{x}}^{\lambda'} (\pi^{\star, \lambda}, \gamma) \big )\\
& \leq \lambda \theta^2 + \sup_{\gamma \in \Gamma} J_{\bm{x}}^{\lambda} (\pi^{\star, \lambda}, \gamma).
\end{split}
\end{equation}
By the optimality of $\pi^{\star, \lambda}$, we have
\[
 V (\bm{x}; \lambda) = \inf_{\pi \in \Pi} \sup_{\gamma \in \Gamma} J_{\bm{x}}^{\lambda} (\pi, \gamma)
= \sup_{\gamma \in \Gamma} J_{\bm{x}}^{\lambda} (\pi^{\star, \lambda}, \gamma).
\]
Therefore, the result follows. 
\end{proof}

\subsection{Proof of Lemma~\ref{lem:lam}}

\begin{proof}
Fix an arbitrary $\lambda  \in  [0, \hat{\lambda})$ and consider any $\lambda' \in (\lambda, \hat{\lambda})$.
Since $\lambda' < \hat{\lambda}$, there exists $t \geq 1$ satisfying $\lambda' \leq \bar{\lambda}_t(\lambda')$, where $\bar{\lambda}_t(\lambda)$ is defined as the maximum eigenvalue of $\Xi^\top P_t^\lambda  \Xi$.
Let $t$ denote the largest time index in $\argmax \{ t \mid \lambda' \leq \bar{\lambda}_t(\lambda')\}$. 
Since $\lambda' > \bar{\lambda}_\tau(\lambda')$ for $\tau = t+1, \ldots, T$, the optimal value functions for $\tau = t, \ldots, T$ are characterized as
$V_\tau (\bm{x}; \lambda') = \bm{x}^\top P_{\tau}^{\lambda'} \bm{x} + 2(r_{\tau}^{\lambda'})^\top \bm{x} + z_\tau^{\lambda'}$ by using the inductive argument in the proof of Theorem~\ref{thm:fin}.

Now consider the optimal value functions with  the penalty parameter $\lambda$.
The Bellman recursion at $t-1$ is given by
\begin{equation}\nonumber
\begin{split}
&V_{t-1} (\bm{x}; \lambda ) =  \bm{x}^\top  Q \bm{x} + \inf_{\bm{u} \in \mathbb{R}^m} \bigg[ \bm{u}^\top R \bm{u}   + \frac{1}{N} \sum_{i=1}^N \sup_{ w\in \mathbb{R}^k} \big  \{ V_{t}(A\bm{x} + B \bm{u} + \Xi w; \lambda)-\lambda \lVert \hat{w}^{(i)}_t - w \rVert^2 \big  \} \bigg].
\end{split}
\end{equation}
Since $\lambda< \lambda'$, we have $V_{t}(\bm{x}; \lambda) \geq V_{t}(\bm{x}; \lambda')$ for all $\bm{x} \in \mathbb{R}^n$. 
Therefore,
\begin{equation}\nonumber
\begin{split}
&V_{t-1} (\bm{x}; {\lambda}) \geq \bm{x}^\top  Q \bm{x} + \inf_{\bm{u} \in \mathbb{R}^m} \bigg[ \bm{u}^\top R \bm{u}   + \frac{1}{N} \sum_{i=1}^N \sup_{ w\in \mathbb{R}^k} \big  \{ V_{t}(A\bm{x} + B \bm{u} + \Xi w; {\lambda'})-\lambda \lVert \hat{w}^{(i)}_t - w \rVert^2 \big  \} \bigg].
\end{split}
\end{equation}
Note that the $w$-dependent part of the inner maximization problem is 
\begin{equation}\nonumber
\begin{split}
\sup_{w \in \mathbb{R}^k} \big \{ & w^\top (\Xi^\top P_t^{\lambda'} \Xi  - \lambda I )w  + 2 [  \Xi^\top P_t^{\lambda'}(A\bm{x} + B \bm{u})+ \Xi^\top r_t^{\lambda'} + \lambda  \hat{w}^{(i)}_t]^\top w  \big \}.
\end{split}
\end{equation}
This is a strictly convex quadratic function with respect to $w$ since   $\lambda < \lambda' \leq \bar{\lambda}_t(\lambda')$.
Thus, the supremum must be $+\infty$ and $V_{t-1}(\bm{x}; \lambda) = +\infty$. It follows from the Bellman recursion that $V_0 (\bm{x}; \lambda) = +\infty$.

Next, we consider the case where $\lambda \in (\hat{\lambda}, \infty)$.
It suffices to show that any $\lambda$ in this range satisfies Assumption~\ref{ass:pen}.
Suppose that $\lambda > \hat{\lambda}$ does not satisfy Assumption~\ref{ass:pen}.
By the definition of $\hat{\lambda}$,  there exists at least one  $\tilde{\lambda} \in [\hat{\lambda}, \lambda)$ that satisfies Assumption~\ref{ass:pen}.
Then, we have
\begin{equation}\label{eq:fin}
\inf_{\pi \in \Pi} \sup_{\gamma \in \Gamma} J_{\bm{x}}^{\tilde{\lambda}} (\pi, \gamma ) = c_2 (\tilde{\lambda}),
\end{equation}
which is finite, since $\tilde{\lambda}$ satisfies Assumption~\ref{ass:pen}.
On the other hand, $\lambda$ does not  satisfy Assumption~\ref{ass:pen}, and thus we can take the largest time index $t$ in $\argmax \{ t \mid \lambda \leq \bar{\lambda}_t (\lambda) \}$ in the same way as in the previous case. 
Switching the role of $(\lambda, \lambda')$ in the previous case to that of $(\tilde{\lambda}, \lambda)$, we deduce that 
\begin{equation}\nonumber
\begin{split}
&V_{t-1} (\bm{x}; \tilde{\lambda}) \geq \bm{x}^\top  Q \bm{x} + \inf_{\bm{u} \in \mathbb{R}^m} \bigg[ \bm{u}^\top R \bm{u}   + \frac{1}{N} \sum_{i=1}^N \sup_{ w\in \mathbb{R}^k} \big  \{ V_{t}(A\bm{x} + B \bm{u} + \Xi w; {\lambda})- \tilde{\lambda} \lVert \hat{w}^{(i)}_t - w \rVert^2 \big  \} \bigg].
\end{split}
\end{equation}
since $V_t(\bm{x}; \tilde{\lambda}) \geq V_t(\bm{x}; \lambda)$.
The supremum must be $+\infty$ since $\tilde{\lambda} < \lambda \leq \bar{\lambda}_t(\lambda)$. 
Thus, we have
 \[
\inf_{\pi \in \Pi} \sup_{\gamma \in \Gamma} J^{\tilde{\lambda}}_{\bm{x}} (\pi, \gamma) = \infty,
\]
which is a contradiction to \eqref{eq:fin}.
Therefore, we conclude that any $\lambda>\hat{\lambda}$ must satisfy Assumption~\ref{ass:pen}.

Finally when $\lambda = \hat{\lambda}$, the value of the objective function can be either finite or infinite depending on the initial states, samples, and system matrices.
However, the boundary condition
 is guaranteed by the monotonically decreasing property of the objective function.
Precisely, $c_1$ should be greater than or equal to $c_2(\hat{\lambda}+\epsilon)$ for any $\epsilon >0$.
\end{proof}

\subsection{Proof of Proposition~\ref{prop:lam}}

\begin{proof}
It is clear that $\lambda_*$ minimizes the objective function of \eqref{opt_lambda} in the range $(\hat{\lambda}, \infty)$.
Moreover, it follows from Lemma~\ref{lem:lam} that $\inf_{\pi \in \Pi} \sup_{\gamma \in \Gamma} J^{\lambda}_{\bm{x}} (\pi, \gamma)=\infty$ for  all $\lambda < \hat{\lambda}$. Thus, $\lambda_*$ minimizes the objective function in the range $[0, \hat{\lambda}) \cup (\hat{\lambda}, \infty)$  and it suffices to show that $\hat{\lambda} \theta^2 + c_1 \geq \lambda_* \theta^2 + c_2(\lambda_*)$.

Suppose that $\hat{\lambda} \theta^2 + c_1 <\lambda_* \theta^2 + c_2(\lambda_*)$.
Let
$$
\Delta := (\lambda_* \theta^2 + c_2(\lambda_*)) - (\hat{\lambda} \theta^2 + c_1) > 0.
$$
Since $\lambda_*$ is a minimizer of \eqref{eq:min},
\[
 \lambda_* \theta^2 + c_2(\lambda_*) \leq (\hat{\lambda}+\epsilon)\theta^2 + c_2(\hat{\lambda}+\epsilon) \quad  \forall \epsilon >0,
\]
which is equivalent to
\[
c_1 + \Delta \leq \epsilon\theta^2 + c_2(\hat{\lambda}+\epsilon) \quad \forall \epsilon >0.
\]
Now, let $\epsilon$ be  sufficiently small so that $\epsilon \theta^2 < \Delta$.
Then,  $c_1< c_2(\hat{\lambda}+\epsilon)$, which is a contradiction to the boundary condition in Lemma~\ref{lem:lam}.
Thus, we conclude that  $\hat{\lambda} \theta^2 + c_1 \geq \lambda_\star \theta^2 + c_2(\lambda_\star)$ and the result follows.

We now show that the optimal value functions $V_t(\bm{x}; \lambda) = \bm{x}^\top P_t^\lambda \bm{x} + 2(r_t^\lambda)^\top \bm{x} + z_t^\lambda$ is jointly convex in $(\lambda, \bm{x}) \in  (\hat{\lambda}, \infty)  \times \mathbb{R}^n$ using mathematical induction. 
For $T$, it is clear that $V_T(\bm{x}; \lambda) = \bm{x}^\top Q_f \bm{x}$ satisfies the joint convexity. 
Suppose now that the induction hypothesis is valid for $t$. 
Recall that the Bellman equation for $t-1$ is given by
\begin{equation}\nonumber
\begin{split}
&V_{t-1}(\bm{x}; \lambda) = \bm{x}^\top Q \bm{x} + \inf_{\bm{u} \in \mathbb{R}^m} \bigg [
\bm{u}^\top R \bm{u} \\
&+ \frac{1}{N} \sum_{i=1}^N \sup_{w \in \mathbb{R}^k}
\big \{
V_t (A\bm{x} + B\bm{u} + \Xi w; \lambda) - \lambda \| \hat{w}_t^{(i)} - w\|^2
\big \}
\bigg ].
\end{split}
\end{equation}
For each $w \in \mathbb{R}^k$, $V_t (A\bm{x} + B\bm{u} + \Xi w; \lambda) - \lambda \| \hat{w}_t^{(i)} - w\|^2$ is convex in  $(\lambda, \bm{x})  \in  (\hat{\lambda}, \infty)  \times \mathbb{R}^n$.
Thus, the convexity is preserved through the point-wise supremum, and 
$\bm{u}^\top R \bm{u} 
+ \frac{1}{N} \sum_{i=1}^N \sup_{w \in \mathbb{R}^k}
 \{
V_t (A\bm{x} + B\bm{u} + \Xi w; \lambda) - \lambda \| \hat{w}_t^{(i)} - w\|^2
 \}$ is jointly convex in $(\lambda, \bm{x}, \bm{u})$  on  $(\hat{\lambda}, \infty)  \times \mathbb{R}^n \times \mathbb{R}^m$, which is a convex set. 
Thus, its infimum over $\bm{u} \in \mathbb{R}^m$ is convex in  $(\lambda, \bm{x})  \in  (\hat{\lambda}, \infty)  \times \mathbb{R}^n$.
This completes our mathematical induction, and the result follows. 
\end{proof}

\subsection{Proof of Theorem~\ref{thm:gc}}

\begin{proof}
Fix an arbitrary infinite-horizon policy $\gamma \in \Gamma_{\mathcal{D}}$.
Let $\gamma_{0: T-1}$ be defined as the marginal of $\gamma$ from stage $0$ and $T-1$.
Then, $\gamma_{0: T-1}$  is an admissible policy of the opponent in the finite-horizon setting. 
Therefore, 
\[
\tilde{J}_{\bm{x}} (\pi_{ss}^{\star, \lambda}, \gamma_{0: T-1}; h^\lambda) \leq \sup_{\gamma \in \Gamma_{\mathcal{D}}} \tilde{J}_{\bm{x}} (\pi_{ss}^{\star, \lambda}, \gamma; h^\lambda ),
\]
where, with a slight abuse of notation,
 $\pi_{ss}^{\star, \lambda}$ represents a finite-horizon policy using $\pi_{ss}^{\star, \lambda}$ at every stage,  and
$\tilde{J}_{\bm{x}} (\pi, \gamma; h) :=
 \frac{1}{T} \mathbb{E}^{\pi, \gamma}
\big [
\sum_{t=0}^{T-1} (x_t^\top Q x_t + u_t^\top R u_t ) + h (x_T) \mid x_0 = \bm{x}
\big ]$. 

Let $\tilde{J}_{\bm{x}}^\lambda (\pi, \gamma; h) :=
 \frac{1}{T} \mathbb{E}^{\pi, \gamma}
\big [
\sum_{t=0}^{T-1} (x_t^\top Q x_t + u_t^\top R u_t - \lambda W_2(\mu_t, \nu)^2) + h (x_T) \mid x_0 = \bm{x}
\big ]$.
Using the argument used in the proof of Lemma~\ref{lem:bd}, we can deduce that
\begin{equation}
\begin{split}
\sup_{\gamma \in \Gamma_{\mathcal{D}}} \tilde{J}_{\bm{x}} (\pi_{ss}^{\star, \lambda}, \gamma; h^\lambda ) &\leq \inf_{\lambda' \geq 0} \sup_{\gamma \in \Gamma} \big ( \lambda' \theta^2 + \tilde{J}_{\bm{x}}^{\lambda'} (\pi_{ss}^{\star, \lambda}, \gamma; h^\lambda) \big )\\
& \leq \lambda \theta^2 + \sup_{\gamma \in \Gamma} \tilde{J}_{\bm{x}}^\lambda (\pi_{ss}^{\star, \lambda}, \gamma; h^\lambda ).
\end{split}
\end{equation}
Combining the two inequalities above yields
\begin{equation} \nonumber
\begin{split}
\tilde{J}_{\bm{x}, \infty} (\pi_{ss}^{\star, \lambda}, \gamma; h^\lambda) &=
\limsup_{T \to \infty} \tilde{J}_{\bm{x}} (\pi_{ss}^{\star, \lambda}, \gamma_{0: T-1}; h^\lambda)\\
& \leq \lambda \theta^2  + \limsup_{T \to \infty} \sup_{\gamma \in \Gamma} \tilde{J}_{\bm{x}}^\lambda (\pi_{ss}^{\star, \lambda}, \gamma; h^\lambda )\\
&\leq \lambda \theta^2 + \rho (\lambda),
\end{split}
\end{equation}
where the last inequality follows from Theorem~\ref{thm:inf} (a).
Since $\gamma$ was arbitrarily chosen from $\Gamma_\mathcal{D}$, the first bound holds. 
The second bound can be obtained using the same argument. 
\end{proof}

\subsection{Proof of Lemma~\ref{lem:gc}}

\begin{proof}
There exists $\hat{\lambda}_{11}>0$ satisfying 
\[
\Phi := BR^{-1} B^\top - \Xi \Xi^\top / \lambda \succeq 0\quad \forall \lambda >\hat{\lambda}_{11},
\]
 since $B R^{-1} B^\top \succ 0$.
It follows from the stabilizability of $(A, B)$ and the observability of $(A, C)$ that the ARE of the standard LQG has a unique PSD solution $\tilde{P}$.
Moreover, the LQG control gain $\tilde{K}:=- R^{-1} B^\top (I + \tilde{P}BR^{-1}B^\top)^{-1} \tilde{P}A$ stabilizes the closed-loop system, such that
\begin{equation}\nonumber
\begin{split}
A+B\tilde{K} &= A - B R^{-1} B^\top (I + \tilde{P}BR^{-1}B^\top)^{-1} \tilde{P}A
\end{split}
\end{equation}
is stable and all eigenvalues of this matrix lie inside the unit circle.
Then, there exists $\hat{\lambda}_{12}$ such that all eigenvalues of
\begin{equation}\nonumber
\begin{split}
&A-(B R^{-1}B^\top - \Xi \Xi^\top/\lambda) (I + \tilde{P}BR^{-1}B^\top)^{-1} \tilde{P}A\\
& = A-\Phi (I + \tilde{P}BR^{-1}B^\top)^{-1} \tilde{P}A := A + \Phi K'
\end{split}
\end{equation}
lie inside the unit circle for any $\lambda>\hat{\lambda}_{12}$, since $B R^{-1}B^\top - \Xi \Xi^\top/\lambda$ is continuous in $\lambda > 0$ and converges to $B R^{-1}B^\top$ as $\lambda \to +\infty$.
Since $A+ \Phi K'$ is stable for any $\lambda>\hat{\lambda}_{12}$, we can conclude that $(A, \Phi^{1/2} )$ is stabilizable for any $\lambda>\hat{\lambda}_{12}$.
Letting $\hat{\lambda}_1 := \max_{t=1,2}\{\hat{\lambda}_{1t}\}$, the result follows.
\end{proof}

\subsection{Proof of Proposition~\ref{prop:gc}}

\begin{proof}
Fix $\lambda \in  (\hat{\lambda}_\infty, \infty)$. Then, $\lambda > \hat{\lambda}_1$ and Assumption~\ref{ass:W} holds by Lemma~\ref{lem:gc}.
Moreover, Assumption~\ref{ass:pen} also holds since $\lambda > \hat{\lambda}_2$.
Thus, Assumptions~\ref{ass:pen}--\ref{ass:ob} hold, and the steady-state average cost $\rho(\lambda)$ exists as defined in Proposition~\ref{prop:avgcost}.
Recall that $\rho(\lambda) = \limsup_{T \to \infty} \min_{\pi \in \Pi} \max_{\gamma \in \Gamma} J_{\bm{x}, T}^\lambda (\pi, \gamma)$ and $J_{\bm{x}, T}^\lambda (\pi, \gamma)$ is monotonically decreasing with respect to $\lambda$ for any $T > 0$. 
Thus, $\rho$ is a monotonically nonincreasing function.
The limit \eqref{lqgavgcost} directly follows from the definition of $\rho(\lambda)$ in Proposition~\ref{prop:avgcost}.

We now show that $\rho(\lambda)$ is convex on $(\hat{\lambda}_\infty, \infty)$ using the convexity result in the finite-horizon case. 
Fix any $\lambda_1, \lambda_2 \in (\hat{\lambda}_\infty, \infty)$ and $\alpha \in (0, 1)$.
We then have
\begin{equation}\nonumber
\begin{split}
\alpha \rho(\lambda_1) + (1-\alpha) \rho(\lambda_2)
& = \limsup_{T \rightarrow \infty} \frac{\alpha}{T} V (\bm{x}; \lambda_1)+ \limsup_{T \rightarrow \infty} \frac{1-\alpha}{T} V (\bm{x}; \lambda_2)\\
& \geq \limsup_{T \rightarrow \infty}  \frac{1}{T} \big [ \alpha V (\bm{x}; \lambda_1) + (1-\alpha) V (\bm{x}; \lambda_2) \big ]\\
& \geq \limsup_{T \rightarrow \infty}  \frac{1}{T} V (\bm{x}; \alpha \lambda_1+ (1-\alpha) \lambda_2)  = \rho(\alpha \lambda_1 + (1-\alpha) \lambda_2),
\end{split}
\end{equation}
where the last inequality comes from the convexity of $V (\bm{x}; \lambda)$ shown in Proposition~\ref{prop:lam}.
Therefore, we conclude that $\rho(\lambda)$ is convex on $(\hat{\lambda}_\infty, \infty)$.
\end{proof}

\subsection{Proof of Theorem~\ref{thm:out}}

\begin{proof}
If the true probability measure $\mu_t$ is contained in the Wasserstein ambiguity set $\mathcal{D}$ for all $t = 0, 1, \ldots, T-1$, 
it follows from Theorem~\ref{thm:bd} that
\begin{equation} \nonumber
\begin{split}
&\frac{1}{T} \mathbb{E}^{\pi^\star_{\hat{w}}}_{w \sim \mu} \big[C_T(x, u) \mid x_0 = \bm{x} \big] \leq \lambda \theta^2 +   V (\bm{x}; \lambda) \quad \forall \bm{x} \in \mathbb{R}^n.
\end{split}
\end{equation}
Therefore, the probability of the expected cost being no greater than  $ \lambda \theta^2 +  V(\bm{x}; \lambda)$  is greater than or equal to the probability that $\mu_t \in \mathcal{D}$ for all $t = 0, 1, \ldots, T-1$.
We then have
\begin{equation*}
\begin{split}
    &\mu^N \bigg\{  \hat{w} :  \frac{1}{T} \mathbb{E}^{\pi^\star_{\hat{w}}}_{w \sim \mu} \big[C_T(x, u) \mid x_0 = \bm{x} \big]  \leq \lambda \theta^2 +  V (\bm{x}; \lambda) \;\; \forall \bm{x} \in \mathbb{R}^n \bigg \}\\
    &\geq \prod_{t=0}^{T-1} \mu_t^N \big \{ \hat{w}_t : W_2(\mu_t,\nu_{\hat{w}_t})^2 \leq \theta (N, \beta)^2 \big \}\\
    &\geq \big  (1-c_1 [b_1(N, \theta^2) \mathbf{1}_{\{ \theta^2 \leq 1\}} + b_2(N, \theta^2) \mathbf{1}_{\{ \theta^2 > 1\}}] \big )^T.
\end{split}
\end{equation*}
The radius $\theta (N, \beta)$ stated in the theorem satisfies
$1-\beta = (1-c_1 [b_1(N, \theta^2) \mathbf{1}_{\{\theta^2 \leq 1\}} + b_2(N, \theta^2) \mathbf{1}_{\{\theta^2 > 1\}}])^T$, and therefore the probabilistic  guarantee \eqref{eq:out} holds.
\end{proof}

\bibliographystyle{IEEEtran}

\bibliography{reference}

\end{document}